\documentclass[3p]{elsarticle}
\usepackage{ecrc}
\usepackage{lineno,hyperref}
\modulolinenumbers[5]
\volume{00}
\firstpage{1}
\runauth{}
\usepackage{amssymb}
\usepackage[figuresright]{rotating}
\usepackage{stfloats}
\usepackage{graphicx}
\usepackage{subfigure}
\usepackage{amsmath}
\usepackage{tabularx}
\usepackage{bbding}
\usepackage{natbib}
\usepackage{url}
\usepackage{comment}
\usepackage{multirow}
\usepackage{float}
\usepackage{epstopdf}
\usepackage[noend]{algpseudocode}
\usepackage{CJK}
\usepackage{algorithm}
\usepackage{algorithmicx}
\begin{document}

\begin{frontmatter}

\title{Practical Two-party Privacy-preserving Neural Network Based on Secret Sharing}

\author[Address1]{Zhengqiang~Ge}
\author[Address1]{Zhipeng~Zhou}
\author[Address2]{Dong~Guo}
\author[Address1]{Qiang~Li \corref{cor1}}

\cortext[cor1]{Corresponding author.\\
	E-mail address: li\_qiang@jlu.edu.cn.\\}

\address[Address1]{College of Computer Science and Technology, Jilin University, Changchun 130012, China}

\address[Address2]{Key Laboratory of Symbolic Computation and Knowledge Engineering of Ministry of Education, Jilin University, Changchun, China}

\begin{abstract}
Neural networks, with the capability to provide efficient predictive models, have been widely used in medical, financial, and other fields, bringing great convenience to our lives. However, the high accuracy of the model requires a large amount of data from multiple parties, raising public concerns about privacy. Privacy-preserving neural network based on multi-party computation is one of the current methods used to provide model training and inference under the premise of solving data privacy. In this study, we propose a new two-party privacy-preserving neural network training and inference framework in which privacy data is distributed to two non-colluding servers. We construct a preprocessing protocol for mask generation, support and realize secret sharing comparison on 2PC, and propose a new method to further reduce the communication rounds. Based on the comparison protocol, we construct building blocks such as division and exponential, and realize the process of training and inference that no longer needs to convert between different types of secret sharings and is entirely based on arithmetic secret sharing. Compared with the previous works, our work obtains higher accuracy, which is very close to that of plaintext training. While the accuracy has been improved, the runtime is reduced, considering the online phase, our work is 5$ \times $ faster than SecureML,     4.32-5.75$ \times $ faster than SecureNN, and is very close to the current optimal 3PC implementation, FALCON. For secure inference, as far as known knowledge is concerned, we should be the current optimal 2PC implementation, which is 4-358$ \times $ faster than other works.
\end{abstract}


\end{frontmatter}

\section{Introduction}

Neural networks have made huge breakthroughs in image recognition, speech recognition, recommendation system, and intrusion detection. The improvement of model accuracy mainly relies on a large amount of training data. Usually, these data cannot be provided by only one company or department. Hence, it is necessary to aggregate and process data from multiple parties. However, due to the public’s sensitivity to privacy and the existence of laws within and between countries, such as HIPAA, PCI and GDPR, direct aggregation and processing of data are not allowed. At the commercial level, due to the competitive relationship between companies, the disclosure of data is not only a matter of privacy and security but also involves disputes over commercial interests.

For example, an Internet company in China has a huge amount of basic user information and user portraits, such as online browsing behaviors, and another insurance company has a large number of insurance users. As such, there is room for cooperation between the two companies. The insurance company can use the user behaviors of the Internet company to make many commercial judgments, such as granting insurance or not. However, due to the conflicts of interest between the two companies, the users of the insurance company themselves have commercial value that cannot be shared with the Internet company. The idea is the same for the Internet company. Under this premise, the subsequent data aggregation and cooperation of co-training are impossible. Due to the existence of these difficulties, the technical advantages of neural networks are prevented from being prominently displayed in such specific practices.

Privacy-preserving neural network based on multi-party computation (MPC) is currently one of the main methods to solve this problem, which involves Secret Sharing\cite{shamir1979share}, Garbled Circuit\cite{yao1982protocols}, Oblivious Transfer\cite{peikert2008framework,asharov2013more}, Homomorphic Encryption\cite{bendlin2011semi,damgaard2012multiparty,paillier1999public,lindner2011better,aono2017privacy}, and other cryptographic knowledge. SecureML\cite{mohassel2017secureml}, ABY$^3$\cite{mohassel2018aby3}, SecureNN\cite{wagh2018securenn}, and FALCON\cite{wagh2020falcon} have all performed considerable research in this area and made outstanding contributions to the advancement of this technology. The framework construction of this method involves many situations, such as 2PC (two-party computation)\cite{mohassel2017secureml}, 3PC (three-parties computation)\cite{mohassel2018aby3,wagh2018securenn,wagh2020falcon,araki2016high,furukawa2017high}, and even Multi-PC (multiple-parties computation) \cite{rachuri2019trident}. Although Multi-PC has improved computation or communication efficiency, there are still problems. The existence of multiple parties makes it more difficult to ensure that all participants do not collude with each other, and there are difficulties in the actual deployment. The current mainstream 3PC or Multi-PC security model must guarantee an honest majority, but this requirement cannot be easily satisfied in the real situation. Even if it is deployed in cloud servers belonging to different entities, the collusion of parties cannot be easily controlled. However, 2PC can meet this condition. In the above example, the servers can be separately deployed in the two companies. Because they have a naturally hostile relationship, collusion with each other is contrary to their respective interests, so there is no motive for collusion. Accordingly, 2PC can also be extended to Multi-PC, as long as each party divides the data and sends them to the two servers. However, for 3PC or Multi-PC, this kind of naturally hostile security guarantee does not exist. Therefore, we believe that the 2PC-based privacy-preserving neural network is still the mainstream in this direction.

Currently, the problems faced by this method mainly include the following points: 1) The time efficiency of training models, 2) Lack of model accuracy due to data expression and the approximation of some functions that are unfriendly to MPC, and 3) Security model that the framework can support, which is mainly defined by the Universal Composition Framework\cite{kushilevitz2010information,canetti2000security,canetti2001universally}. The security model here includes the semi-honest and malicious adversary models (The semi-honest model requires the computing participants to strictly follow the protocol and only be curious about the data in the computation process, whereas the malicious adversary model does not restrict the participants to follow the protocol, they can deviate from the protocol and cause computation failure). The currently known two-party privacy-preserving neural network, like SecureML, only supports the semi-honest model, so our framework also follows it.

The time efficiency problem is currently the most important factor hindering the large-scale application of this method. The computation time of the current privacy preserving neural network is still far from the time required for plaintext training, because it includes that this method requires more time-consuming computation under the expression of MPC, and also needs participants to communicate with one another. Moreover, the time overhead caused by the communication cannot be ignored for the whole time. In the Local Area Network (LAN) setting, because the bandwidth is large enough, the communication time can be negligible compared to the computation time. Conversely, in the Wide Area Network (WAN) setting, due to the reduction in bandwidth, the proportion of communication time in the whole time will significantly increase. Hence, reducing the communication overhead can reduce the communication time in the whole process, which is also an important method to reduce the overall time. 

In order to reduce the overall time and communication overhead, of course, we must start with the basic computation. From the micro perspective, neural network includes linear computation and nonlinear computation. The linear computation mainly consists in the matrix multiplication of the fully connected layer and the convolutional layer, that is, addition and multiplication. The nonlinear computation consists in maxpool layers and nonlinear activation functions, such as rectified linear unit (ReLU), Sigmoid, and Softmax. The basic nonlinear computation involved includes comparison, division, and exponential. Generally, Secret Sharing is more efficient for linear computation, among which addition is free, multiplication can be computed by the pre-computed multiplication triplets, but it is less efficient for nonlinear computation or even cannot be computed, whereas Garbled Circuit can express these functions. However, for linear computation, the circuit has high depth and low efficiency. At the same time, Garbled Circuit needs to introduce a large amount of communication, and the garbler needs to express the entire computation as a circuit and send the truth table to the decoder. Consequently, the huge amount of communication will also reduce the efficiency. Regarding the advantages of the two cryptographic methods, the main research direction of the current study is to mix the two methods: the linear part is computed by Secret Sharing, and the nonlinear part is computed by Garbled Circuit. However,    the conversion of the two types of secret also brings considerable overheads. ABY\cite{demmler2015aby}, SecureML and ABY$^3$ have optimized a lot in this regard. For emphasis, in the computation of the neural network, the matrix multiplication still accounts for more of the whole process, so the most important thing is to express more nonlinear computation as a way to compute using Secret Sharing. While reducing the number of calling for Garbled Circuit, the overhead of secret conversion can also be further reduced.

Regarding the lack of precision, because some functions are not easy to express in Garbled Circuit or the time efficiency caused by the expression is particularly poor, it needs to be handled by approximate polynomial functions. The computation of polynomials only includes addition and multiplication, so Secret Sharing can be competent, but there is still a gap between the approximate function and original function, resulting in a lack of accuracy in the results. Ultimately, a compromise is required between the model accuracy and time efficiency. Furthermore, because neural network requires high-precision float-point numbers but Secret Sharing is performed in the fixed-point number domain, it is necessary to convert float-point numbers into fixed-point numbers. In order to ensure the precision, we need transform the number by letting ${x'} = {2^{{l_D}}}x$, making the least most ${l_D}$ bits of the fixed-point number represent the fractional part of the float-point number. For the further analysis on the lack of precision, please refer to SecureML\cite{mohassel2017secureml}. Because enough decimals can be retained, the lack of precision in this part can actually be ignored.

Based on the current difficulties mentioned above, we design a 2PC-based faster and modular privacy-preserving neural network framework based on secret sharing. The main contributions of this study are as follows:
\begin{itemize}
  \item [1)] 
The pre-processing protocol supporting 2PC secret sharing comparison. The secret sharing comparison protocol has been implemented in SecureNN, but because the generation of its mask r requires a neutral third party, it is built on 3PC. In order to support the application of this method to 2PC, we propose a new preprocessing protocol, use the OT protocol to generate the mask ${\left\{ {{{\left\langle {r[j]} \right\rangle }^p}} \right\}}$ over  $Z_p$, get $r_0$ and $r_1$, and use the garbled circuit to compute $wrap\left( {{r_0},{r_1},L} \right)$, which makes the protocol applicable on 2PC.
  \item [2)]
Secret sharing comparison on 2PC. We use the ideas in FALCON and based on the pre-processing protocol to construct the secret sharing comparison protocol on 2PC, but its idea of string multiplication when checking 0 is still not optimal. We propose a new method to reduce the number of communication rounds from the original $\log _2^l$ (where $l$ is the data length, $l=64$ in this study) to 3 rounds, further reducing the overall online computation time. At the same time, reducing the calls for secret sharing multiplication can also reduce the number of pre-computed multiplication triplets.
  \item [3)]
Other secret sharing nonlinear protocols. Based on the secret sharing comparison protocol on 2PC, we implement the secret sharing division. At the same time, we use a piecewise linear function approximation to achieve $e^x$, and based on the two functions, we realize the Softmax function entirely based on Secret Sharing, which is more practical in multi-classification tasks. The experiment shows that as long as the function is sufficiently detailed in the specified domain of definition, the same training effect as the real function can be achieved.
  \item [4)]	
We give the specific implementation of the framework and conduct experiments of neural network training and inference on the same four network structures as the previous works. Our experiment results show that our framework has achieved higher accuracy. The training accuracy of every network structure is higher than the current optimal result, which is roughly the same as the result of plaintext training, especially for the SecureML network. The results obtained by the previous work only reaches 93.4\%, but we get 94.8\% close to the plaintext training. In terms of training time, considering the online phase, in the LAN setting, our work is $5\times$ faster than SecureML, $4.32-5.75\times$ faster than SecureNN, and is very close to FALCON, and in the WAN setting, we get the approximate results. For inference time, in terms of known knowledge, we should be the current optimal 2PC implementation, which is $4-358\times$ faster than others.
\end{itemize}

\subsection {Paper Structure}
\label{sec:2}

In \textbf{Section 2}, we discuss some related work. In \textbf{Section 3}, we give the preliminaries about this paper. In \textbf{Section 4}, we present all the basic building blocks that support neural network training and inference, In \textbf{Section 5}, we give the security analysis of the corresponding algorithm, and in \textbf{Section 6} we summarize the use of all the building blocks. In \textbf{Section 7}, we provide detailed experimental data and corresponding performance evaluation. Finally, in \textbf{Appendix}, we give the ideal functionality descriptions.

\section {Related work}
\label{sec:3}

The related research on privacy-preserving neural network based on MPC mainly includes training and inference. The training part requires to ensure the privacy of training data and privacy of the result model, and the inference part requires to ensure that the data holder cannot explore the knowledge of the model and the model holder cannot have any understanding of the data required to be predicted. Because the training part includes forward propagation and back propagation and the inference part only includes forward part, the content is relatively simple and the related work is also relatively extensive\cite{riazi2018chameleon,rouhani2018deepsecure,juvekar2018gazelle,li2020optimizing,chandran2017ezpc,hesamifard2017cryptodl,ball2019garbled,ma2019non,chabanne2017privacy,gilad2016cryptonets,hesamifard2017cryptodl,liu2017oblivious,rouhani2018deepsecure,riazi2019xonn}. CryptoNet\cite{gilad2016cryptonets} is the earliest implementation of privacy-preserving inference based on Homomorphic Encryption, using squares instead of the sigmoid function and mean pooling instead of max pooling. CryptoDL\cite{hesamifard2017cryptodl} further improves CryptoNet using low-order polynomials to approximate nonlinear functions. MiniONN\cite{liu2017oblivious} proposes the use of the single instruction, multiple data (SIMD) batch to improve the pre-computed secret sharing protocol and the method of approximating the nonlinear activation function with polynomials. Without changing the trained model, the model is converted into an oblivious form, achieving a privacy-preserving inference. Chameleon\cite{riazi2018chameleon} uses a hybrid protocol, which uses GMW\cite{goldreich2019play} for low-order nonlinear functions, Garbled Circuit for high-order nonlinear functions, Secret Sharing for linear functions to achieve inference, and proposes improvements to vector computation and a new method for multiplication triplets generation, but the selection of a neutral third party is not easy in the actual operation. In a follow-up work, DeepSecure\cite{rouhani2018deepsecure} designs the optimized realization of the components in a neural network based on Garbled Circuit. XONN\cite{riazi2019xonn} also uses Garbled Circuit, but the matrix multiplication is replaced by the free XNOR gate. Particularly, it changes the weights in the neural network from floating point numbers to $\left\{ { - 1,0,1} \right\}$. Both of the above maintain the same security guarantee, that is, the parameters of the neural network model are used as privacy, and the composition of the model, the number of layers, and the number of neurons in each layer can be open to the demander.

The training part of the neural network involves a large number of derivation operations due to the back-propagation phase and includes nonlinear functions, such as Softmax. Compared with the inference part, the training part is more difficult and has less research\cite{mohassel2017secureml,mohassel2018aby3,wagh2018securenn,agrawal2019quotient,lou2019glyph,wagh2020falcon,gupta2018distributed,zheng2020industrial,li2020npmml}. The earliest research to realize the neural network training is SecureML, which improves the generation of multiplication triplets, uses Secret Sharing for linear computation, Garbled Circuit for nonlinear computation, such as comparison, and uses ReLU to construct the approximate functions of Sigmoid and Softmax. Then, the training of neural network is realized for the first time, but the time and accuracy of the model still have much room for improvement. QUOTIENT\cite{agrawal2019quotient} tailors the weight to $\left\{ { - 1,0,1} \right\}$, uses MPC to construct the quantization and normalization in the back propagation of the neural network, and designs the optimization algorithm in the fixed-point domain based on the adaptive gradient optimization. ABY$^3$ improves the conversion among arithmetic sharing, Boolean sharing, and Yao sharing and realizes the training of neural networks on 3PC for the first time. Based on the basic framework of SecureML and adding a third-party assistant, SecureNN proposes a new comparison protocol based on Secret Sharing and builds secret sharing ReLU and Maxpool functions based on this. FALCON further constructs a 3PC neural network training method with equal status of the three parties on the basis of SecureNN and ABY$^3$ and improves the comparison protocol and linear computation for 3PC, which has greatly improved time efficiency and framework security. However, the improved comparison protocol still requires the number of multiplications proportional to the data bits, so there is still room for improvement. Moreover, the three recent papers are all research on 3PC. To the best of our knowledge, only a few studies on the neural network training of 2PC have been performed in the past two years.

\section{Preliminaries}
\label{sec:4}

\subsection{Neural Network}
\label{sec:6}

Neural network consists of several types of layer, such as the fully connected layer, convolutional layer, pooling layer, and dropout layer. Each layer contains several neurons, which contain weights and biases, and the corresponding nonlinear activation functions, such as 

\begin{equation}
 \begin{aligned}
ReLU: f\left( x \right) = \max \left( {x,0} \right)(comparison)
 \end{aligned}
 \label{equ:1}
\end{equation}

\begin{equation}
 \begin{aligned}
 Sigmoid: f\left( x \right) = {1 \over {1 + {e^{ - x}}}}(exponential, division)
 \end{aligned}
 \label{equ:2}
\end{equation}
\begin{equation}
 \begin{aligned}
Softmax:f\left( x \right) =  {{{e^{ - {u_i}}}} \over {\sum\nolimits_{i = 1}^{dm} {{e^{ - {u_i}}}} }}(comparison, exponential, division)
 \end{aligned}
 \label{equ:3}
\end{equation}
Training a neural network using the stochastic gradient descent (SGD) method can be divided into two stages: forward propagation and back propagation. The purpose of forward propagation is to compute the difference between the predicted result and the actual label. Assuming that there are ${d_i}$ neurons in the ${i^{th}}$ layer of the neural network, the process of forward propagation can be formalized as 
\begin{equation}
 \begin{aligned}
{X_i} = f\left( {{X_{i - 1}} \cdot {w_i} + {b_i}} \right)
 \end{aligned}
 \label{equ:4}
\end{equation}
where $f$ is the nonlinear activation function, ${{w_i}}$ is the weight, and ${{b_i}}$ is the bias. Then, we compute the loss function according to the result of the output layer. The loss function in the neural network is generally the cross-entropy function: 
\begin{equation}
 \begin{aligned}
{C_i}\left( w \right) =  - {y_i}\log y_i^* - (1 - {y_i})\log (1 - y_i^*)
 \end{aligned}
 \label{equ:5}
\end{equation}
where $y_i^*$ is the result of the output layer. The back propagation aims to update the weight and bias of each neuron according to the loss function, the process is mainly based on the chain rule 
\begin{equation}
 \begin{aligned}
 {Y_i} = ({Y_{i + 1}} \times W_i^T) \cdot {{\partial f\left( {{U_i}} \right)} \over {\partial {U_i}}}
 \end{aligned}
 \label{equ:6}
\end{equation}
Then, it updates the parameters according to the derivative result 
\begin{equation}
 \begin{aligned}
{w_i}: = {w_i} - {\alpha  \over {\left| B \right|}} \cdot {X_i} \times {Y_i}
 \end{aligned}
 \label{equ:7}
\end{equation}
where $\alpha $ is the learning rate and $B$ is the batch. In the neural network, the basic computations mainly involved are matrix multiplication, comparison, exponential and division. For each layer, two matrix multiplications are required, and comparison and exponential are specifically applied according to the selection of the neuron's nonlinear activation function. Therefore, the optimization of comparison and matrix multiplication is of great significance to the improvement of the time efficiency of the whole neural network.

\subsection{Secure Computation}
\subsubsection{Secret Sharing}
\label{sec:5}

In our whole protocol, all intermediate data are shared between the two computing servers in the form of arithmetic sharing, and $\left\langle  \circ  \right\rangle $ represents the secret form of a number. Assuming that the data owner holds the data $a$, in order to share $a$, it randomly generates ${a_0} \in {Z_{{2^l}}}$, computes and gets ${a_1} = a - {a_0}\bmod {2^l}$. Then, it can send  $\left\langle a \right\rangle _0$ and $\left\langle a \right\rangle _1$ to the two servers $P_0$ and $P_1$ separately for specific computation. Because $a_0$ and $a_1$ are random relative to the original $a$, the privacy of the data will not be leaked to the two actual computing servers, and secret sharing addition and multiplication will become relatively easy. For addition, suppose that $a = {\left\langle a \right\rangle _0} + {\left\langle a \right\rangle _1}$, $b = {\left\langle b \right\rangle _0} + {\left\langle b \right\rangle _1}$, and $P_0$, $P_1$ hold ${\left\langle a \right\rangle _0},{\left\langle b \right\rangle _0}$ and ${\left\langle a \right\rangle _1},{\left\langle b \right\rangle _1}$ respectively. To get $c = a + b$, $P_i$ only needs to locally compute ${c_i} = {a_i} + {b_i},i \in \{ 0,1\} $, because ${c} = {c_0} + {c_1} = {a_0} + {a_1} + {b_0} + {b_1}$, in which addition is free and does not need communication in secret sharing. For multiplication,  it needs to precompute the multiplication triplets $\left\langle x \right\rangle ,\left\langle y \right\rangle $ and $\left\langle z \right\rangle $ where $z = x \cdot y$, and the specific protocol is presented as \textbf{Algorithm \ref{algorithm1}}. To perform a secret sharing multiplication, a set of multiplication triplets is required, and the party needs to send two messages $\left\langle e \right\rangle $ and $\left\langle f \right\rangle $ to the other party, making the linear computation in the previous works require many multiplication triplets and increasing communication.

\begin{CJK*}{UTF8}{gkai}
    \begin{algorithm}
    \caption{Secret Sharing Multiplication}
    \normalsize\textbf{Input}: ${P_i}$ holds ${\left\langle a \right\rangle _i},{\left\langle b \right\rangle _i}$
    \\\textbf{Output}: ${P_i}$ gets ${\left\langle c \right\rangle _i}$ ,where ${\left\langle c \right\rangle _i} = {\left\langle a \right\rangle _i}*{\left\langle b \right\rangle _i}$
    \\\textbf{Common Randomness}: ${P_i}$ holds one pair of multiplicaiton triplets $({x_i},{y_i},{z_i})$
    \\
        \begin{algorithmic}[1] 
                \State for $i \in \{ 0,1\} $ ${P_i}$ do:
                \State \qquad ${\left\langle e \right\rangle _i} = {\left\langle a \right\rangle _i} - {\left\langle x \right\rangle _i},{\left\langle f \right\rangle _i} = {\left\langle b \right\rangle _i} - {\left\langle y \right\rangle _i}$
                \State \qquad reveal $e$ and $f$
                \State \qquad ${P_i}$ gets ${\left\langle c \right\rangle _i} = f \times {\left\langle x \right\rangle _i} + e \times {\left\langle y \right\rangle _i} + {\left\langle z \right\rangle _i} + i \times e \times f$
                \State end for
        \end{algorithmic}
      \label{algorithm1}
    \end{algorithm}
\end{CJK*}	

\subsubsection{Oblivious Transfer and Garbled Circuit}
Oblivious transfer (OT) is a fundamental cryptographic primitive that is commonly used as building blocks in MPC. In the protocol, a sender ${\rm S}$ has two inputs ${x_0}$ and ${x_1}$, and a receiver ${\rm R}$ has a selection bit $b$ and wants to obtain ${x_b}$ without learning anything else or revealing $b$ to ${\rm S}$. The notion $\left( { \bot ;{x_b}} \right) \leftarrow$ OT$\left( {{x_0},{x_1};b} \right)$ can be used to denote the protocol.

Garbled Circuit (GC) is another generic protocol in MPC that requires only a constant number of communication rounds. A garbled circuit protocol consists of a garbling algorithm that generates a garbled circuit ${\rm F}$; a decoding tab $dec$; an encoding algorithm that generates garbled input $\hat x$; an evaluation algorithm takes $\hat x$ and ${\rm F}$ as input and returns the garbled output $\hat z$; and a decoding algorithm that takes the decoding tabel $dec$ and $\hat z$ and return $f(x)$. In our protocol, OT and GC are only used as black boxs in the relevant protocols.

\subsection{Notation}
\label{sec:7}

In our protocol, we use additive secret sharing over the two rings $Z_L$ and $Z_p$, where $L = {2^l}$ and $p$ is a prime. All the numbers involved in the formal computation are additively shared in ring $Z_L$, and in this work, $l = 64$. $Z_p$ is only used in the sercet sharing comparison. To comapre two numbers $a$ and $b$, we need to get the relationship of each bit of them, and to share them  between two parties, each bit of the 64-bit secret is additively shared in $Z_p$,  here we choose $p = 67$. that is,  each 64-bit number is shared as a vector of 64 shares, and each share is a value between 0 and 66. And we denote by $x\left[ j \right]$ the ${j^{th}}$ component of a vector $x$.

The wrap function which will be mentioned many times later is defined as a function of the secret shares of the parties and effectively compute the "carry bit" when the shares are added together as integers. and is formally defined as follows:
\begin{equation}
 wrap({a_1},{a_2},L) = \left\{ {
\begin{matrix}
   0 & {
\begin{matrix}
   {if} & {{a_0} + {a_1} < L}  \cr 
\end{matrix}
 }  \cr 
   1 & {Otherwise}  \cr 
\end{matrix}
 } \right.
\end{equation}

\subsection{Security Model}
We consider the situation that two or more clients who want to train models on their joint data, but do not want to leak the data privacy. We do not make any assumption on the distribution of the data, it can be horizontally, vertically or arbitrarily partitioned among the clients. When there are only two clients, they themselves can be the servers participate in computation, and the data does not need to be shared. Because it can be assumed that the data has been shared, but one party has all the data, and the number of locations corresponding to the other party is 0. While when there are more than two clients, we can select two of them as servers, and other clients secretly share their data and send
them to the two servers separately, and the two server need to be non-colluding.

We assume that there is a semi-honest adversary ${\cal A}$ who can corrupt any subset of the clients and at most one of the two servers, this confirms that the two servers must be non-colluding again. If one of the servers is controlled by the adversary, aonther one must behave honestly. The security definition requires that the adversary only learns the data from the clients and the server it has controlled and the final output but nothing about other clients or another server. We define the security using the Universal Composition framework, and the overview is given in \textbf{Section 5}.

\section{Building Blocks}
\label{sec:8}

\subsection{Preprocessing protocol}
\label{sec:9}

In order to realize the secret sharing comparison on 2PC, we need to generate the secret sharing of the mask $r$, ${\left\langle r \right\rangle _0},{\left\langle r \right\rangle _1}$, the secret sharing of each bit of $r$ over $Z_p$, ${\left\{ {{{\left\langle {r[j]} \right\rangle }^p}} \right\}_i},j \in [0,l),i \in \{ 0,1\} $, and $wrap\left( {{r_0},{r_1},L} \right)$. In SecureNN, these tasks are all computed by a non-colluding neutral third party, and then sent to two participants for specific computation, so the masks for the two parties are completely random. However, it is introduced whether this neutral third party will collude with one of the parties. In order to prevent this from happening, we no longer introduce a neutral third party, and only through the two parties, achieve a more efficient 2PC-based secret sharing comparison protocol. For this reason, we propose a new preprocessing protocol to generate the completely random mask $r$.

For the secret sharing of the corresponding bit of $r$ over $Z_p$, we can use the OT protocol. First, for each bit, $P_0$ randomly generates $a \in (0,p)$ as the secret sharing of the bit, then generates
$\left\{ {{b_0},{b_1}} \right\} = \left\{ {p - a,p - a + 1} \right\}$, and randomly swaps. $P_1$ randomly generates choice bit $c \in \left\{ {0,1} \right\}$, and $P_0,P_1$ jointly execute the OT protocol $( \bot ;{b_c}) \leftarrow$ OT$({b_0},{b_1};c)$. $P_1$ gets $b$ over $Z_p$ corresponding to the bit. Considering the properties of the OT protocol, $P_0,P_1$ do not know if the corresponding value is 0 or 1, so the value is shared. $P_i$ gets $\{ \left\langle {r[j]} \right\rangle _i^p\} ,(i \in \{ 0,1\} )$. With the secret sharing result of the corresponding bit, multiply and add the value corresponding to the bit to get the secret sharing ${\left\langle r \right\rangle _i}$,
\begin{equation}
 \begin{aligned}
{\left\langle r \right\rangle _i} = \sum\limits_{j = 0}^l {pow(2,j)*(r[j] - p/2 - i)}
 \end{aligned}
 \label{equ:8}
\end{equation}
The specific protocol is presented as \textbf{Algorithm \ref{algorithm2}}.

\begin{CJK*}{UTF8}{gkai}
    \begin{algorithm}
    \caption{GenerateMaskR}
    \normalsize\textbf{No Input}
    \\\textbf{Output}: ${P_i}$ gets ${\left\langle r \right\rangle _i}$, and $\{ \left\langle {r[j]} \right\rangle _i^p\} ,(i \in \{ 0,1\} )$
    \\\textbf{No Randomness}
    \\
        \begin{algorithmic}[1] 
                \State for $j \in \{ 0,1,...,l-1\}$ do : 
                \State \qquad ${P_0}$ generate random ${a_j} \in (0,p)$, get ${\{ {b_0},{b_1}\} _j} = {\{ p - a,p - a + 1\} _j}$, and shuffle ${\{ {b_0},{b_1}\} _j}$
                \State \qquad ${P_1}$ generate choice bit ${c_j} \in \{ 0,1\} $
                \State end for
                \State ${P_0},{P_1}$ invoke the OT protocol $( \bot ;{b_c}) \leftarrow OT({b_0},{b_1};c)$
                \State ${P_0}$ gets $\left\{ {{{\left\langle {r[j]} \right\rangle }^p}} \right\} \leftarrow \left\{ {{a_j}} \right\}$ and ${P_1}$ gets $\left\{ {{{\left\langle {r[j]} \right\rangle }^p}} \right\} \leftarrow \left\{ {{b_j}} \right\}$
                \State ${P_0},{P_1}$ locall compute ${\left\langle r \right\rangle _i} = \sum\nolimits_{j = 0}^l {pow(2,j)*(r[j] - p/2 - i)} $ to get ${\left\langle r \right\rangle _i}$
        \end{algorithmic}
    \label{algorithm2}
    \end{algorithm}
\end{CJK*}

As regards whether ${\left\langle r \right\rangle _0} + {\left\langle r \right\rangle _1}$ wraps or not, what we can guess is that it must have a relationship with the most significant bit (MSB), MSB$\left( {{{\left\langle r \right\rangle }_0}} \right)$, MSB$\left( {{{\left\langle r \right\rangle }_1}} \right)$ of the two secret sharings and the highest bit MSB$\left( r \right)$ of $r$. If MSB$\left( {{{\left\langle r \right\rangle }_0}} \right)$,MSB$\left( {{{\left\langle r \right\rangle }_1}} \right)$ both are 0, no matter MSB$\left( r \right)$ is 0 or 1, it will not wrap. If MSB$\left( {{{\left\langle r \right\rangle }_0}} \right)$,MSB$\left( {{{\left\langle r \right\rangle }_1}} \right)$  both are 1, no matter MSB$\left( r \right)$ is 0 or 1, it will wrap. If MSB$\left( {{{\left\langle r \right\rangle }_0}} \right)$,MSB$\left( {{{\left\langle r \right\rangle }_1}} \right)$ are 0, 1, it will not wrap when MSB$\left( r \right)$ is 1, and will wrap when it is 0. MSB$\left( {{{\left\langle r \right\rangle }_0}} \right)$, MSB$\left( {{{\left\langle r \right\rangle }_1}} \right)$ both parties can take out separately, and MSB$\left( r \right)$ is held by both parties, which is $\left\langle {r[0]} \right\rangle _i^p$. We obtain the truth table \textbf{TABLE \ref{tab1}} by exploring the relationship between
 MSB$({\left\langle r \right\rangle _0})$, MSB$({\left\langle r \right\rangle _1})$, $\left\langle {r[0]} \right\rangle _0^p$, $\left\langle {r[0]} \right\rangle _1^p$, where ${m_i} =$ MSB$\left( {{{\left\langle r \right\rangle }_i}} \right)$, ${m^i}$ represents the value corresponding to $\left\langle {r[0]} \right\rangle _i^p$ over $Z_2$, and ${m^i} = \left( {\left\langle {r\left[ 0 \right]} \right\rangle _i^p + i} \right)\% 2,i \in \left\{ {0,1} \right\}$. Through the truth table, we derive the following logical expression:  
\begin{equation}
 \begin{aligned}
wrap({r_0},{r_1},L) = ({m_0}\& {m^0}\& {m^1})|({m_1}\& {m^0}\& {m^1})|({m_0}\& \sim{m^0}\& \sim{m^1})|({m_1}\& \sim{m^0}\&\sim{m^1})|({m_0}\& {m_1})
 \end{aligned}
\end{equation}
Then, we use the garbled circuit for computation, and the specific protocol is presented as \textbf{Algorithm \ref{algorithm3}}.

\begin{table*}[!t]
\centering
\caption{Truth table of $wrap({r_0},{r_1},L)$, where ${m_i}$ represents MSB$({\left\langle r \right\rangle _i})$, and ${m^i}$ represents the value corresponding to $\left\langle {r[0]} \right\rangle _i^p$ over $Z_2$.}
 \setlength{\tabcolsep}{3.8mm}{
\renewcommand{\arraystretch}{1.5}
\begin{tabular}{|c|c|c|c|c|c|c|c|c|c|c|c|c|c|c|c|c|}
\hline
${m_0}$ & 0 & 0 & 0 & 0 & 0 & 0 & 0 & 0 & 1 & 1 & 1 & 1 & 1 & 1 & 1 & 1 \\ \hline
${m_1}$ & 0 & 0 & 0 & 0 & 1 & 1 & 1 & 1 & 0 & 0 & 0 & 0 & 1 & 1 & 1 & 1 \\ \hline
${m^0}$ & 0 & 0 & 1 & 1 & 0 & 1 & 1 & 0 & 0 & 1 & 1 & 0 & 0 & 0 & 1 & 1 \\ \hline
${m^1}$ & 0 & 1 & 0 & 1 & 0 & 1 & 0 & 1 & 0 & 1 & 0 & 1 & 0 & 1 & 0 & 1 \\ \hline
$wrap$ & 0 & 0 & 0 & 0 & 1 & 1 & 0 & 0 & 1 & 1 & 0 & 0 & 1 & 1 & 1 & 1 \\ \hline
\end{tabular}}
\label{tab1}       
\end{table*}

\begin{CJK*}{UTF8}{gkai}
    \begin{algorithm}
    \caption{Get Wrapped}
    \normalsize\textbf{No Input}
    \\\textbf{Output}: ${P_i}$ gets $wrap({r_0},{r_1},L)$
    \\\textbf{No Randomness}
    \\
        \begin{algorithmic}[1] 
                \State for $i \in \{ 0,1\} $ ${P_i}$ do : 
                \State \qquad ${m_i} \leftarrow MSB({r_i})$ , ${m^i} \leftarrow (r[0]_i^p + i)\% 2$
                \State end for
                \State ${P_0},{P_1}$ invoke the Garbled Circuit Protocol and get 
                \State \qquad $wrap({r_0},{r_1},L) = $$({m_0}\& {m^0}\& {m^1})|({m_1}\& {m^0}\& {m^1})|({m_0}\& $\~{}${m^0}\& $\~{}${m^1})|({m_1}\& $\~{}${m^0}\& $\~{}${m^1})|({m_0}\& {m_1})$
        \end{algorithmic}
        \label{algorithm3}
    \end{algorithm}
\end{CJK*}

\subsection{Secret Sharing Comparison}
\label{sec:10}

As the basis of the neural network, privacy-preserving comparison has always been a hot research. The most well-known comparison is the Millionaire Problem, where two millionaires want to know who has more money but they do not want the other to know how much money they have. This is a comparison problem in which neither party wants to disclose privacy. The first solution proposed is Garbled Circuit proposed by Yao[30]. However, due to the low computational efficiency caused by the huge communication of the circuit, many researchers have also proposed other schemes for comparison in neural network. The comparison of $a$ and $b$ is actually to obtain MSB of $\left( {a - b} \right)$. When the MSB is 0, then $a > b$ and the result is opposite when the MSB is 1. To determine the MSB of the arithmetic secret sharing, SecureNN proposes to use the idea of MSB$\left( a \right) =$ LSB$\left( {2a} \right)$ over the odd ring. Here, the data are first converted to an odd ring, and then judging the least significant bit (LSB) becomes simple. FALCON further developed this idea and considered that the real MSB can be obtained from the addition of the MSB of the three arithmetic secret sharing and the carry bit of the previous value. However, these methods are only applicable to 3PC, and there is no secret sharing comparison protocol for 2PC.
We combine SecureNN and FALCON to get a secret sharing comparison protocol for 2PC. The main idea is to evaluate the MSB of the secret sharing results of the difference and convert it into the XOR result of the MSB of the two secret sharing results ${c_0},{c_1}$ and the bit representing whether the previous value is wrapped or not: 

\begin{equation}
 \begin{aligned}
\left\langle a \right\rangle  > \left\langle b \right\rangle ? = MSB\left( {\left\langle a \right\rangle  - \left\langle b \right\rangle } \right) = MSB\left( {{c_0}} \right) \oplus MSB\left( {{c_1}} \right)\oplus wrap\left( {2{c_0},2{c_1},L} \right)
 \end{aligned}
 \label{equ:10}
\end{equation}
where $\left\langle c \right\rangle  = \left\langle a \right\rangle  - \left\langle b \right\rangle $. As regards wrapping, it can use the pre-generated random number $r$ to mask the data and attribute it to the comparison of the random r and a plaintext $x$. The pre-generated secret sharing of $r$ over $Z_p$ is compared with the plaintext $x$ in order of bits. If $r > x$, then there must be a certain bit position $k$, all the digits before the bit satisfy that $r\left[ j \right] = x\left[ j \right]$, which is $r\left[ j \right] \oplus x\left[ j \right] = 0$, $j \in \{ k+1,k+2,...l-1\}$. At this bit, $r\left[ k \right] - x\left[ k \right] = 1$, $- \left( {r\left[ k \right] - x\left[ k \right]} \right) =  - 1$
\begin{equation}
\begin{aligned}
c\left[ k \right] = 1 - \left( {r\left[ k \right] - x\left[ k \right]} \right)\sum\nolimits_{j = k + 1}^{l-1} {\left( {x\left[ j \right] \oplus r\left[ j \right]} \right)}  = 0
 \end{aligned}
\label{equ:11}
\end{equation}
Convert the problem to find whether there is a 0 in $\left\{ {{{\left\langle {c\left[ j \right]} \right\rangle }^p}} \right\}$,$j \in \left\{ {0,1,...l - 1} \right\}$. Then we can get the final result. But the problem is that the position of 0 and each value in $\left\{ {{{\left\langle {c\left[ j \right]} \right\rangle }^p}} \right\}$ cannot be got by both parties. Once the parties know the position of 0, the approximate size of $x$ can be inferred. Knowing the values of $\left\{ {{{\left\langle {c\left[ j \right]} \right\rangle }^p}} \right\}$, it can determine the value of $x$ based on the change of $c\left[ j \right]$ from front to back and the $r\left[ j \right]$ it holds. SecureNN adopts the method of shuffling all  $\left\{ {{{\left\langle {c\left[ j \right]} \right\rangle }^p}} \right\}$, and multiplying with the same value to cover up, and then sending them to a neutral third party for judgment, but this is the same as the above problem, which is likely to cause collusion with one party. FALCON uses the method of string multiplication to solve this problem. This method can cover all the values of $\left\{ {{{\left\langle {c\left[ j \right]} \right\rangle }^p}} \right\}$ and only get the final result whether there is 0. But we think this method is far from optimal. Because it needs to multiply the data on each bit together, even if the parallel computation is used, it still needs to be executed in order of $\log _2^l - 1$ times of multiplication, which has $\log _2^l - 1$ communication rounds, plus one communication to reveal the plaintext, a total of $\log _2^l$ communication rounds are required. The total number of communication rounds is large. The communication time plus the mutual waiting time result in a total longer computation time. Meanwhile, on 2PC, computing multiplication requires pre-computation of multiplication triplets, which consumes a lot of offline time. Therefore, we propose a new method to check 0 with fewer communication rounds, which is presented as \textbf{Algorithm \ref{algorithm4}}.

\begin{CJK*}{UTF8}{gkai}
    \begin{algorithm}
    \caption{CheckZero}
    \normalsize\textbf{Input}: ${P_i}$ holds $\{ \left\langle {c[j]} \right\rangle _i^p\} ,i \in \left\{ {0,1} \right\},j \in \left\{ {0,1,...,l - 1} \right\}$
    \\\textbf{Output}: ${P_0}$ gets $\eta $
    \\\textbf{Common Randomness} ${P_0}$, ${P_1}$ hold the same random seed $\theta $
    \\
        \begin{algorithmic}[1] 
                \State ${P_0},{P_1}$  randomly shuffle $\{ \left\langle {c[j]} \right\rangle _i^p\} $  using seed $\theta $
                \State for $j$ $\in \left\{ {0,1,...,l/2 - 1} \right\}$, ${P_0},{P_1}$ compute ${\left\langle {d\left[ j \right]} \right\rangle ^p} = {\left\langle {c\left[ j \right]} \right\rangle ^p} \cdot {\left\langle {c\left[ {j + l/2} \right]} \right\rangle ^p}$
                \State ${P_0}$ randomly generates $N \in \left\{ {1,2,...,p - 1} \right\}$, for $j \in \left\{ {0,1,...,l/2 - 1} \right\}$ computes $\left\langle {{d^*}[j]} \right\rangle _0^p = \left\langle {d[j]} \right\rangle _0^p \cdot N\bmod p$, and sends $\left\{ {\left\langle {{d^*}[j]}  \right\rangle _0^p} \right\}$ to ${P_1}$
                \State ${P_1}$ receives $\left\{ {\left\langle {{d^*}[j]} \right\rangle _0^p} \right\}$ and for $j \in \left\{ {0,1,...,l/2 - 1} \right\}$ generates ${M_j} \in \left\{ {1,2,...,p - 1} \right\}$ computes $\left\langle {{d^{**}}[j]} \right\rangle _0^p = \left\langle {{d^*}[j]} \right\rangle _0^p \cdot {M_j}\bmod p,\left\langle {{d^*}[j]} \right\rangle _1^p = \left\langle {d[j]} \right\rangle _1^p \cdot {M_j}\bmod p$, shuffle $\left\{ {\left\langle {{d^{**}}[j]} \right\rangle _0^p} \right\},\left\{ {\left\langle {{d^*}[j]} \right\rangle _1^p} \right\}$ synchronously and randomly, and sends $\left\{ {\left\langle {{d^{**}}[j]} \right\rangle _0^p} \right\},\left\{ {\left\langle {{d^*}[j]} \right\rangle _1^p} \right\}$ to ${P_0}$
                \State ${P_0}$ receives $\left\{ {\left\langle {{d^{**}}[j]} \right\rangle _0^p} \right\}$,$\left\{ {\left\langle {{d^*}[j]} \right\rangle _1^p} \right\}$, for $j\in\left\{ {0,1,...,l/2 - 1} \right\}$ computes $\left\langle {{d^{**}}[j]} \right\rangle _1^p = \left\langle {{d^*}[j]} \right\rangle _1^p \cdot N\bmod p$ and gets $\left\{ {\left\langle {{d^{**}}[j]} \right\rangle _0^p} \right\},\left\{ {\left\langle {{d^{**}}[j]} \right\rangle _1^p} \right\}$
                \State for $j \in \left\{ {0,1,...,l/2 - 1} \right\}$ ${P_0}$ computes $d[j] = \left\langle {{d^{**}}[j]} \right\rangle _0^p + \left\langle {{d^{**}}[j]} \right\rangle _1^p\bmod p$. If it exists j, $d\left[ j \right] = 0$ then $\eta  = 1$, else $\eta  = 0$. ${P_0}$ gets $\eta $
   
        \end{algorithmic}
      \label{algorithm4}
    \end{algorithm}
\end{CJK*}

In \textbf{Algorithm \ref{algorithm4}}, the problem we need to solve is that both parties find whether there is 0 in $\left\{ {{{\left\langle {c\left[ j \right]} \right\rangle }^p}} \right\}$ without knowing the values of $\left\{ {{{\left\langle {c\left[ j \right]} \right\rangle }^p}} \right\}$ and the location of 0. If we do not want the two parties to know $\left\{ {{{\left\langle {c\left[ j \right]} \right\rangle }^p}} \right\}$, we need to mask all of them. We know that over $Z_p$, multiplying a random number (Not 0) will change its value, but the result of multiplying 0 by any number is still 0, so we can use multiplication to cover up the result. In addition, in order to prevent both parties from knowing the location of 0, we can use that one party cannot reveal the data but can shuffle them, while the other party does not know the process of shuffle, that is, does not know the location of 0, but can reveal the data (The data here are masked by multiplication) and find 0 to solve this problem. Here let $P_0$ is responsible for finding 0, and $P_1$ is responsible for shuffling and masking. Because $\left\{ {{{\left\langle {c\left[ j \right]} \right\rangle }^p}} \right\}$ are secretly shared between $P_0$ and $P_1$, in order to synchronize shuffling and masking, the part of data of $P_0$ needs to be sent to $P_1$ first. There is a problem here. If they are sent directly, the data can be revealed, which does not meet our expectations. Then $P_0$ also needs to mask its own data. If $P_0$ multiplies each ${c\left[ j \right]}$ by different random numbers, when the other party shuffles, masks and sends them back, because it does not know the shuffle process, it will not know each ${c\left[ j \right]}$ and correspondence of the random number multiplied, the original data cannot be revealed. If each ${c\left[ j \right]}$ is multiplied by the same random number, when sent to the other party, beacuse this random number is over $Z_p$ and the domain space is small, only $p-1$ computations are needed to get all possible results, and then $\left\{ {{{\left\langle {c\left[ j \right]} \right\rangle }^p}} \right\}$ can be revealed. According to the relationship between each ${c\left[ j \right]}$, we can know that  $\left\{ {{{\left\langle {c\left[ j \right]} \right\rangle }^p}} \right\}$ should have an approximately increasing relationship, and the difference between the two numbers before and after is at most 1 or 2, then we can reveal the correct result and get the value and the position of 0. Therefore, here we adopt the method of first shuffling  $\left\{ {{{\left\langle {c\left[ j \right]} \right\rangle }^p}} \right\}$ based on the same random number on both sides, and then multiplying the two. Although both parties know the aforementioned relationship that exists in $\left\{ {{{\left\langle {c\left[ j \right]} \right\rangle }^p}} \right\}$, they do not know the specific value of each bit. Multiplying them after shuffling will destroy this relationship, and there is no longer an approximately increasing relationship. $P_0$ multiplies them by the same random number $N \in \left\{ {1,2,...,p - 1} \right\}$ and sends them to $P_1$. Even if it can get all possible results, it cannot determine the real data, but each pair of ${c\left[ j \right]}$ can still maintain the original corresponding relationship.
$P_1$ can randomly shuffle each pair of ${c\left[ j \right]}$, multiply them by  different random numbers ${M_j} \in \left\{ {1,2,...,p - 1} \right\}$, and send them back. Then $P_0$ multiplies the part of the data of $P_1$ with the same random number $N$, and then add the corresponding values together to find whether there is 0. Since $P_0$ does not know the shuffle process of $P_1$, $P_0$ will not know the original position corresponding to 0, and each  ${c\left[ j \right]}$ obtained is the result of multiplying $M_j$ of $P_1$, which is not the same as the original value. The CheckZero protocol is presented as \textbf{Algorithm \ref{algorithm4}}. Combining with it, we get the 2PC secret sharing comparison protocol, which is presented as \textbf{Algorithm \ref{algorithm5}}.

\subsubsection{Efficiency Discussion}
\label{sec:11}

For each comparison, if the problem of checking 0 is solved according to the string multiplication, for the communication, because the multiplication is performed over $Z_p$,  the data length only needs 1 byte. Then, for string multiplication, the total communication required is $\sum\nolimits_{i = 1}^{\log _2^l - 1} {4 \cdot {l \over {{2^i}}}}  + 2 = \sum\nolimits_{i = 1}^{\log _2^l - 1} {{l \over {{2^{i - 2}}}} + 2}$ bytes. While our method requires $4 \cdot {l \over 2} + {l \over 2} + l = 3.5l$ bytes. As regards communication rounds, string multiplication requires $\log _2^l - 1$ times of multiplication, and then the final reveal, totally  $\log _2^l$ rounds of communication are required. Our method requires only one multiplication, $P_0$ sends messages to $P_1$, and then $P_1$ sends messages to $P_0$, totally 3 communication rounds. We assume here that the data length $l=64$, then we will reduce the number of communication rounds from 6 to 3. At the same time, since the number of multiplications is reduced from $\sum\nolimits_{i = 1}^{\log _2^{l - 1}} {{l \over {{2^i}}}} $ to ${l \over 2}$, the number of multiplication triplets required is also reduced.

\begin{CJK*}{UTF8}{gkai}
    \begin{algorithm}
    \caption{Secret Sharing Comparison}
    \normalsize\textbf{Input}: ${P_i}$ holds ${\left\langle a \right\rangle _i}$
    \\\textbf{Output}: ${P_i}$ gets $bit(a > 0)$
    \\\textbf{Common Randomness}: ${P_i}$ holds ${\left\langle r \right\rangle _i}$, $\{ \left\langle {r[j]} \right\rangle _i^p\} $(shares of bits of r), and the bit $\alpha$ where $\alpha = wrap({r_0},{r_1},L)$
    \\
        \begin{algorithmic}[1] 
                \State for $i \in \{ 0,1\} $ ${P_i}$ do : 
                \State \qquad Compute ${x_i}=2{a_i}+{r_i}$
                \State \qquad Compute ${\beta _i} = wrap(2{a_i},{r_i},L)$
                \State \qquad Reconstruct $x = {\sum x _i}(\bmod L)$
                \State \qquad Compute $\delta  = wrap({x_0},{x_1},L)$
                \State \qquad for $j \in \{ l-1,l-2,...,0\} $ do:
                \State \qquad \qquad Compute shares of $c[j] = 1 - (x[j] - r[j])+\sum\nolimits_{k = j + 1}^{l-1} {(x[k] \oplus r[k])}$
                \State \qquad Invoke CheckZero and ${P_0}$ gets $\eta $
                \State \qquad Compute $\theta  = {\beta _0} + {\beta _1} + \delta  - \eta  - \alpha $
                \State \qquad Return $bit(a > 0) = MSB({a_0}) \oplus MSB({a_1}) \oplus \theta $
                \State end for
        \end{algorithmic}
       \label{algorithm5}
    \end{algorithm}
\end{CJK*}

\subsection{Other Nonlinear Protocols}

\begin{figure}
\centering
\includegraphics[width=0.7\textwidth]{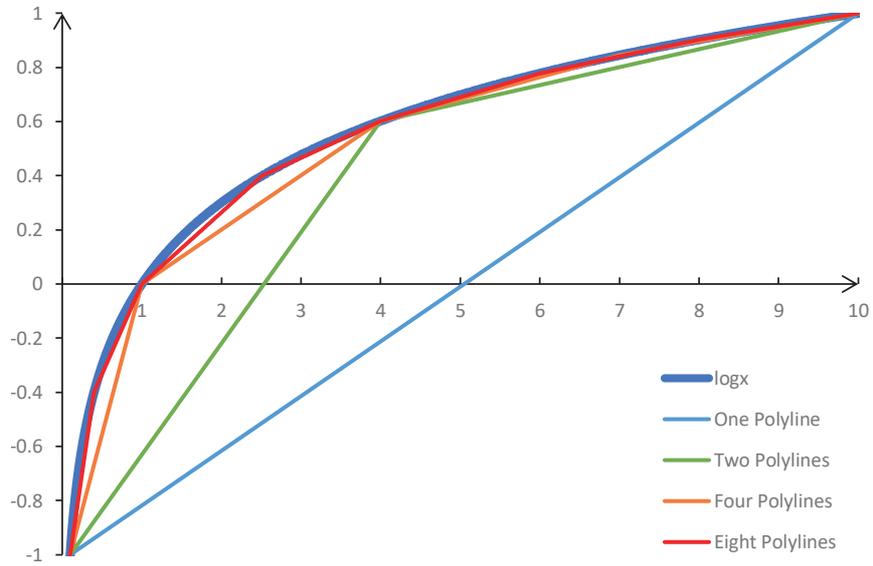}
\caption{Log and its segmentation approximate image}
\label{fig1}
\end{figure}

Similar to the method in FALCON, we implement ReLU, Max, MaxPool, Pow(\textbf{Algorithm \ref{algorithm7}}), and Division \textbf{Algorithm\textbf{ \ref{algorithm8}}}) on 2PC based on the comparison protocol. Because the format of the protocol is roughly the same, it will not be introduced in detail. Here, we give the detailed introduction to the implementation method of the secret sharing Softmax function. As an important nonlinear function in multi-classification tasks, the previous work, as in SecureML, is to approximate it with the ReLU function, 
\begin{equation}
 \begin{aligned}
ASM = {{{\mathop{\rm Re}\nolimits} lu\left( {{u_i}} \right)} \over {\sum {{\mathop{\rm Re}\nolimits} lu\left( {{u_i}} \right)} }}
 \end{aligned}
\end{equation}
Although the model can be trained normally, the accuracy of the model is still lower than that of plaintext training. Since the basic computations involved only include $max$, $e^x$, and $division$. Both $division$ and $max$ have been implemented through the secret sharing comparison algorithm mentioned above. Therefore, as long as we express the exponential function, we can realize the Softmax function based entirely on secret sharing. Based on the reason that the secret sharing method is easy to express the linear computation, we try to approximate $e^x$ with a polynomial. The previous work tries to use a polynomial to approximate the Sigmoid function, but the problem is that when using low-order polynomials, it has bad approximation effect and huge error, while when using higher-order polynomials, it leads to low efficiency, and the higher-order polynomials cannot fit every part of the function. We use the piecewise polynomial method to divide the function into sufficiently detailed parts in a finite domain. Each part is replaced by a linear function. Just like when we were learning derivatives, the teacher once said that as long as the segmentation is detailed enough, then this small line segment can replace the original function, which is presented as \textbf{Figure \ref{fig1}}. For $e^x$, since $x_j-M<0$ (M = Max$(x)$), the domain of $e^x$ is $\left( { - \infty ,0} \right)$, but we can only approximate the function on a finite field.  Because ${e^{ - 10}} < {10^{ - 3}}$, the process error of computing the probability in Softmax is already small enough, so we approximate the domain of $e^x$ to $\left( { -10 ,0} \right)$, divide the specified domain into finite parts, and judge the domain of $x$ by the secret sharing comparison algorithm, then the result can be obtained by bringing in the specified linear function. The detailed level of the specific segmentation is given in the follow-up experiment, and the secret sharing exponential algorithm is presented as \textbf{Algorithm \ref{algorithm6}}. With it, we can get the Softmax function entirely based on sercet sharing.

\begin{CJK*}{UTF8}{gkai}
    \begin{algorithm}
    \caption{Secret Sharing Exponential}
    \normalsize\textbf{Input}: ${P_i}$ holds ${\left\langle x \right\rangle _i}$
    \\\textbf{Output}: ${P_i}$  gets  ${\left\langle {{e^{{x_0} + {x_1}}}} \right\rangle _i}$ 
    \\\textbf{No Randomness}
    \\
        \begin{algorithmic}[1] 
                \State for $i \in \{ 0,1\} $ ${P_i}$ do:
                \State \qquad split the specified domain $({x_0},{x_n})$ into $({x_0},{x_1}),({x_1},{x_2}).....({x_{n - 1}},{x_n})$
                \State \qquad compute the value y corresponding to each division point $y = {e^{{x_j}}}$
                \State \qquad compute the slope of the function on each segment of the domain ${k_j} = {{{y_{j + 1}} - {y_j}} \over {{x_{j + 1}} - {x_j}}}$, get n sets of linear functions $y - {y_j} = {k_j}(x - {x_j})$
                \State \qquad invoke Secret Sharing Comparison and get ${j}$, where $x \in \left( {{x_j},{x_{j + 1}}} \right),j \in \left[ {0,n - 1} \right]$
                \State \qquad invoke Secret Sharing Add/Multiplication and get $\left\langle y \right\rangle $ corresponding to $\left\langle x \right\rangle $
                \State end for
        \end{algorithmic}
       \label{algorithm6}
    \end{algorithm}
\end{CJK*}

\begin{CJK*}{UTF8}{gkai}
    \begin{algorithm}
    \caption{Secret Sharing Pow}
    \textbf{Input}: ${P_i}$ holds ${\left\langle b \right\rangle _i}$
    \\\textbf{Output}: ${P_i}$  gets $\alpha ,{2^{\alpha  - 1}} < b \le {2^\alpha }$
    \\\textbf{No Randomness}
    \\
        \begin{algorithmic}[1] 
                \State for $i \in \{ 0,1\} $ ${P_i}$ do:
                \State \qquad $\left\langle x \right\rangle  = \left\langle b \right\rangle ,\alpha  = 0$
                \State \qquad for $i \in \{ l - 1.....,2,1\} $ do:
                \State \qquad \qquad $\left\langle {{d_x}} \right\rangle  = \left\langle x \right\rangle  - {2^{{2^i} + \alpha }}$
                \State \qquad \qquad invoke Secret Sharing Compare and get $c = bit({d_x} > 0)$
                \State \qquad \qquad if $c = 1$ then:
                \State \qquad \qquad \qquad $\left\langle x \right\rangle  = \left\langle {{d_x}} \right\rangle ,\alpha  = \alpha  + {2^i}$
                \State \qquad \qquad end if
                \State \qquad  end for
                \State \qquad ${P_i}$ gets $\alpha $
                \State end for
        \end{algorithmic}
       \label{algorithm7}
    \end{algorithm}
\end{CJK*}

\begin{CJK*}{UTF8}{gkai}
    \begin{algorithm}
    \caption{Secret Sharing Division}
    \textbf{Input}: ${P_i}$ holds ${\left\langle a \right\rangle _i},{\left\langle b \right\rangle _i}$
    \\\textbf{Output}: ${P_i}$  gets $\left\langle {a/b} \right\rangle $ with a given fixed presion ${f_p}$
    \\\textbf{No Randomness}
    \\
        \begin{algorithmic}[1] 
                \State for $i \in \{ 0,1\} $ ${P_i}$ do:
                \State \qquad invoke Secret Sharing Pow and get $\alpha $ where ${2^{\alpha  - 1}} < b \le {2^\alpha }$
                \State \qquad $c = b/{2^{\alpha}}$
                \State \qquad ${w_0} = 2.9142 - 2c$
                \State \qquad compute ${\varepsilon _1} = 1 - c \cdot {w_0}$ and ${\varepsilon _1} = \varepsilon _0^2$ and ${\varepsilon _2} = \varepsilon _1^2$
                \State \qquad ${P_i}$ gets ${\left\langle {{a \over b}} \right\rangle _i} = a{w_0}(1 + {\varepsilon _0})(1 + {\varepsilon _1})(1 + {\varepsilon _2})$
                \State end for
        \end{algorithmic}
       \label{algorithm8}
    \end{algorithm}
\end{CJK*}

\subsection{Linear Computation}

The linear computation in the neural network is mainly concentrated on the fully connected layer and the convolutional layer. The fully connected layer can be computed by matrix multiplication, and the convolutional layer can also be transformed and computed by matrix multiplication too. Here we also use the matrix multiplication form of secret sharing multiplication mentioned in SecureML, the specific algorithm is presented as \textbf{Algorithm \ref{algorithm9}}. Secret sharing matrix multiplication can reduce the use of pre-computed multiplication triplets, and the same mask data is used for the same data. Similarly, for the computational characteristics of neural networks, because it usually trains a set of dataset for multiple epochs, it will cause multiple repetitions of the input layer data. We can make the two servers send their own part of the data masked to each other at the beginning, which can reduce the communication of the input layer and further reduce the use of multiplication triplets.

\begin{CJK*}{UTF8}{gkai}
    \begin{algorithm}
    \caption{Secret Sharing Matrix Multiplication}
    \textbf{Input}: ${P_i}$ holds $\left\langle A \right\rangle _i^{m \times n},\left\langle B \right\rangle _i^{n \times v}$
    \\\textbf{Output}: ${P_i}$ gets $\left\langle C \right\rangle _i^{m \times v}$
    \\\textbf{Common Randomness}: ${P_i}$ holds multiplication triplets $\left\langle U \right\rangle _i^{m \times n},\left\langle V \right\rangle _i^{n \times v}$,$\left\langle Z \right\rangle _i^{m \times v}$
    \\
        \begin{algorithmic}[1] 
                \State for $i \in \{ 0,1\} $ ${P_i}$ do:
                \State \qquad ${\left\langle E \right\rangle _i} = {\left\langle A \right\rangle _i} - {\left\langle U \right\rangle _i},{\left\langle F \right\rangle _i} = {\left\langle B \right\rangle _i} - {\left\langle V \right\rangle _i}$
                \State \qquad Reveal E and F
                \State \qquad ${P_i}$ gets ${\left\langle C \right\rangle _i} = i \cdot E \times F + {\left\langle U \right\rangle _i} \times F + E \times {\left\langle V \right\rangle _i} + {\left\langle Z \right\rangle _i}$
                \State end for
        \end{algorithmic}
       \label{algorithm9}
    \end{algorithm}
\end{CJK*}

\section{Security Analysis}

We use the real world-ideal world simulation paradigm method to prove the security of the protocol. This paradigm is briefly introduced as follows. In a real interaction, the parties execute the protocol $\Pi $ in a certain environment ${\cal Z}$, where an adversary ${\cal A}$ exists, and there is another ideal interaction ${\cal I}$. In the ideal environment ${\cal I}$, all parties send their inputs to a trusted third party to implement the protocol ${\cal F}$ completely and truly. Finally, to prove the security of the protocol, for each adversary ${\cal A}$ that exists in the real interaction ${\cal R}$, there is a simulator ${\cal S}$ in the ideal interaction, if the environment cannot distinguish between the two interactions, this protocol is secure. That is, the information obtained by the adversary in the real interaction and the information obtained by the simulator in the ideal interaction are the same in the category of informatics and are indistinguishable. As a proof, it is only necessary to check whether the designed simulator has the ability to generate messages that are indistinguishable from the real-world interaction messages. Due to space constraints, we formally describe the functionalities in Appendix A. We describe simulators for ${\Pi _{GenerateR}}$ (\textbf{Figure \ref{fig8}}), ${\Pi _{GetWrapped}}$ (\textbf{Figure \ref{fig9}}), ${\Pi _{CheckZero}}$ (\textbf{Figure \ref{fig10}}), ${\Pi _{Compare}}$ (\textbf{Figure \ref{fig11}}), ${\Pi _{e^x}}$ (\textbf{Figure \ref{fig12}}) that achieve indistinguishability. ${{\cal F}_{Mult}}$,${{\cal F}_{OT}}$,${{\cal F}_{GC}}$ are identical to prior works. Thus, our protocol is secure under this paradigm.

\newtheorem{thm}{\bf Theorem}
\begin{thm}\label{thm1}
GenerateMaskR in \textbf{Algorithm \ref{algorithm2}} securely realizes ${{\cal F}_{GenerateMaskR}}$  in the presence of a semi-honest admissible adversary in the ${{\cal F}_{OT}}$  hybrid model.
\end{thm}

\newtheorem{proof}{Proof}[section]
\begin{proof}
The fisrt four steps are computed locally, the numbers are all generated randomly, and do not need to be simulated. The interaction only exists in the OT protocol in Step 5, the simulator for ${{\cal F}_{OT}}$ can be used to simulate the transcripts in it, the distribution of ${b_0},{b_1}$ and $c$ are all uniformly random from the adversary's view.
\end{proof}

\begin{thm}\label{thm2}
GetWrapped in \textbf{Algorithm \ref{algorithm3}} securely realizes ${{\cal F}_{GetWrapped}}$ in the presence of a semi-honest admissible adversary in the ${{\cal F}_{GC}}$ hybrid model.
\end{thm}

\begin{proof}
The first three steps of the algorithm are computed locally, so there is no interaction and no simulation is required. There is interaction only in Step 4 where the garbled circuit is called. and can be simulated by the simulator using in ${{\cal F}_{GC}}$.
\end{proof}

\begin{thm}\label{thm3}
CheckZero in \textbf{Algorithm \ref{algorithm4}} securely realizes ${{\cal F}_{checkzero}}$ in the presence of a semi-honest admissible adversary in the ${{\cal F}_{Mult}}$ hybrid model.
\end{thm}

\begin{proof}
Step 2 calls the secret sharing multiplication, which can be simulated by using the simulator for ${{\cal F}_{Mult}}$. In Step 3,4, the numbers generated are all randomly selected. The simulator can generate transcripts with the same distribution. The result sent back to ${P_0}$ in Step 5 is the same in distribution as the original data, 0 is no longer in the original position, and other values are also the results masked, so the output are uniformly random from the adversary's view.
\end{proof}

\begin{thm}\label{thm4}
Secret Sharing Comparison in \textbf{Algorithm \ref{algorithm5}} securely realizes ${{\cal F}_{compare}}$ in the presence of a semi-honest admissible adversary in the (${{\cal F}_{GenerateMaskR}}$,  ${{\cal F}_{GetWrapped}}$,  ${{\cal F}_{CheckZero}}$) hybrid model.
\end{thm}

\begin{proof}
The common randomness can be simulated using the simulators for ${{\cal F}_{GenerateMaskR}}$, ${{\cal F}_{GetWrapped}}$. The algorithm only communicates in steps 4 and 8, and \textbf{Algorithm \ref{algorithm4}} is called in Step 8. Other steps are all local computation and do not need simulation. In Step 4, the $x$ sent to each other is the result of the mask of the mask $r$ that is randomly generated in advance, from the adversary's view, these transcripts are all uniformly random values, and the simulator $\cal S$ can simulate them in the same distribution. And Step 8 can be simulated using the simulator for ${{\cal F}_{CheckZero}}$.
\end{proof}

\begin{thm}\label{thm5}
Secret sharing exponential in \textbf{Algorithm \ref{algorithm6}} securely realizes ${{\cal F}_{{e^x}}}$ in the presence of a semi-honest admissible adversary in the (${{\cal F}_{Mult}}$,${{\cal F}_{compare}}$) hybrid model.
\end{thm}

\begin{proof}
The algorithm only calls \textbf{Algorithm \ref{algorithm5}} in the step 5 and calls the secret-sharing multiplication in the step 6. Simulation is done as before using the hybrid argument. The protocol simply composes ${{\cal F}_{Mult}}$ and ${{\cal F}_{compare}}$ and hence is simulated using the corresponding simulators.
\end{proof}

Algorithm 1,7,8 and 9 have been proven to be secure in SecureML\cite{mohassel2017secureml} and FALCON\cite{wagh2020falcon}, we only convert them to 2PC. Due to space limitations, we do not introduce them here.

\section{Summary}

\begin{figure*}
\includegraphics[width=\textwidth]{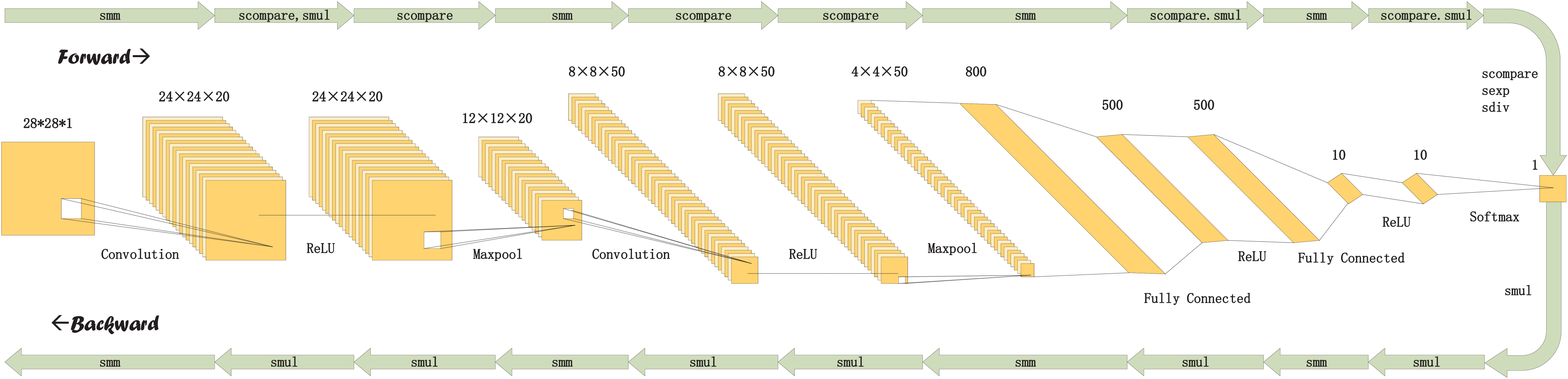}
\caption{LeNet's structure and its secure training flow chart}
\label{fig2}
\end{figure*}

After completing all the basic operations of the neural network training, we present the whole neural network training process in this part. As shown in \textbf{Figure \ref{fig2}}, we give the network structure of LeNet corresponding to Network-C in the subsequent experimental part and the flow chart of the training using the secret sharing method. By only using Secret Sharing, this one of the cryptography methods, we can achieve the whole process of secure training.

\section{Experimental Evaluation}

We give the experimental results of the effect of the exponential function segmentation degree on accuracy, the performance of training and inference on four neural network structures based on the new building blocks using the MNIST dataset\cite{lecun-mnisthandwrittendigit-2010}, and compare with the frameworks designed in the previous works.

\subsection{Experimental Setup}

Our experiment is written in C++ and use the Libtorch and Emp-toolkit library for the invocations of Oblivious Transfer and Garbled Circuit. In the whole protocol, the data length is set to 64, and the retained decimal digit ${l_D}$ is set to 16 for precision. For secret sharing multiplication on 2PC, which requires pre-computation of multiplication triples, we still use the algorithms in SecureML. We run our experiments on a workstation running Ubuntu 18.04 equipped with two GTX 1080Ti graphics cards with 64G RAM in the LAN and WAN setting. For the LAN setting, the bandwidth is approximately 625MBps, and for the WAN setting, we use the traffic control (TC) command to set the port speed limit and give a bandwidth setting of 40MBps.

\subsection{Network Structure Setting}

In order to facilitate the performance comparison, we use the same four neural network structures as in SecureNN, which are also the networks used in recent works, like SecureML, ABY$^3$, and FALCON. Here is a brief introduction,

\begin{itemize}
  \item [1)] 

Network-A: a 3-layer Deep Neural Network mentioned in SecureML\cite{mohassel2017secureml}. (1) A 28$ \times $28 fully connected layer, and the activation functions following this layer are ReLU. (2) A 128 fully conneted layer, and the activation fuctions following this layer are ReLU. (3) A 128 fully connected layer, the activation functions are ReLU, and Softmax is used to get a probability distribution.

  \item [2)]

Network-B: a 4-layer Convolutional Neural Network from MiniONN\cite{liu2017oblivious}.(1) A 2-dimensional convolutional layer with 1 input channel, 16 output channels and a 5$ \times $5 filter. The activation functions following this layer are ReLU, followed by a 2$ \times $2 Maxpool. (2) A 2-dimensional convolutional layer with 16 input channels, 16 output channels and 5$ \times $5 filter.The activation functions following this layer are ReLU and a 2$ \times $2 Maxpool followed. (3) A 256$ \times $100 fully connected layer, the activation functions are ReLU.(4) A 100$\times $10 fully connected layer, the activation functions are ReLU, and Softmax is used to get a probability distribution.

  \item [3)]

Network-C: a 4-layer Convolutional Neural Metwork called LeNet network\cite{lecun1998gradient}. (1) A 2-dimensional convolutional layer with 1 input channel, 20 output channels and a 5$ \times $5 filter. The activation functions following this layer are ReLU, followed by a 2$ \times $2 Maxpool. (2) A 2-dimensional convolutional layer with 20 input channels, 50 output channels and another 5$ \times $5 filter. The activation functions following this layer are ReLU and a 2$ \times $2 Maxpool followed. (3) An 800$ \times $500 fully connected layer. The activation functions are ReLU. (4) A 500$ \times $10 fully connected layer, the activation functions are ReLU, and Softmax is used to get a probability distribution.

  \item [4)]	

Network-D: a 3-layer Convolutional Neural Metwork from Chameleon\cite{riazi2018chameleon}. (1)  a 2-dimensional convolutional layer with a 5$ \times $5 filter, stride of 2, and 5 output channels. The activation functions are ReLU. (2) A 980 fully conneted layer, and the activation fuctions following this layer are ReLU. (3) A 100 fully connected layer, the activation functions are ReLU, and Softmax is used to get a probability distribution. 

\end{itemize}

\subsection{The Effect of Approximating $e^x$ On Accuracy}

\begin{figure}
\centering
\includegraphics[width=0.8\textwidth]{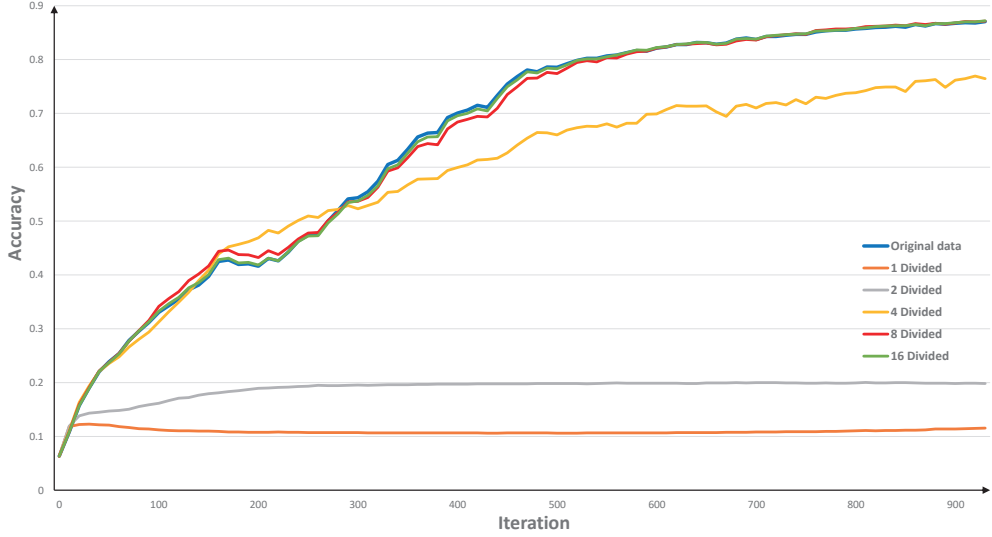}
\caption{The training accuracy when $e^x$ is divided into 1, 2, 4, 8, 16 segments evenly in $\left( { - 10,0} \right)$}
\label{fig3}
\end{figure}

Here we give the effect on the accuracy of model training under different segmentation degrees of $e^x$ in the first two epochs of Network-A. We give the change trend of model training accuracy when $e^x$ is divided into 1, 2, 4, 8, 16  segments evenly in the domain $\left( { - 10,0} \right)$ in \textbf{Figure \ref{fig3}}. It can be understood from the change of the image that with the increase in the detail degree of the segmentation in the specified domain, the accuracy of the secure training is gradually close to the accuracy of the plaintext training. When it is divided into 8 parts, it has been approximately fitted, and when it is divided into 16 parts, 
there is basically no difference in training accuracy with plaintext, which is the most fundamental reason why we can improve the accuracy. The subsequent experimental results are all obtained on the premise that $e^x$ is equally divided into 16 parts in $\left( { - 10,0} \right)$.

\subsection{Secure Training}

We train the model in the LAN and WAN setting based on the four types of networks using the MNIST dataset and obtained the corresponding accuracy, training time, and communication. The setting of the learning rate follows the previous work. The learning rate of Network-A is set to ${2^{ - 7}}$, and those of the other networks are set to ${2^{ - 5}}$. For the training time, we test 100 complete forward- and back-propagation training processes and get the average value to estimate the overall training time of 15 epochs (7020 iterations). And we actually run the training process of 15 epochs once, the actual training time is roughly the same as the estimated time. The method of testing communication is also similar.

\subsubsection{Accuracy Comparison}

\begin{table}
\centering
\caption{Comparison of the training communication of the three networks. All communication is reported in TB.}
\setlength{\tabcolsep}{8mm}{
\renewcommand{\arraystretch}{1.5}
\begin{tabular}{ccccc}
\hline\noalign{\smallskip}
Network & Epoch & SecureNN & Ours & Plaintext  \\
\noalign{\smallskip}\hline\noalign{\smallskip}
A & 15 & 93.40\% & 94.80\% & 94.83\% \\ \hline
\noalign{\smallskip}
B &  5 & 97.94\% & 98.39\% & 98.41\% \\
  & 10 & 98.05\% & 98.93\% & 98.96\% \\
  & 15 & 98.77\% & 99.09\% & 99.13\% \\ \hline
\noalign{\smallskip}
C &  5 & 98.15\% & 98.61\% & 98.60\% \\
  & 10 & 98.43\% & 98.94\% & 98.98\% \\
  & 15 & 99.15\% & 99.17\% & 99.17\% \\ \hline
\noalign{\smallskip}
D & 15 &         & 97.60\% & 97.64\% \\
\noalign{\smallskip}\hline
\end{tabular}}
\label{tab2}       
\end{table}

\begin{figure*}
\centering
\begin{minipage}[c]{0.5\textwidth}
\centering
\includegraphics[height=5cm,width=8.5cm]{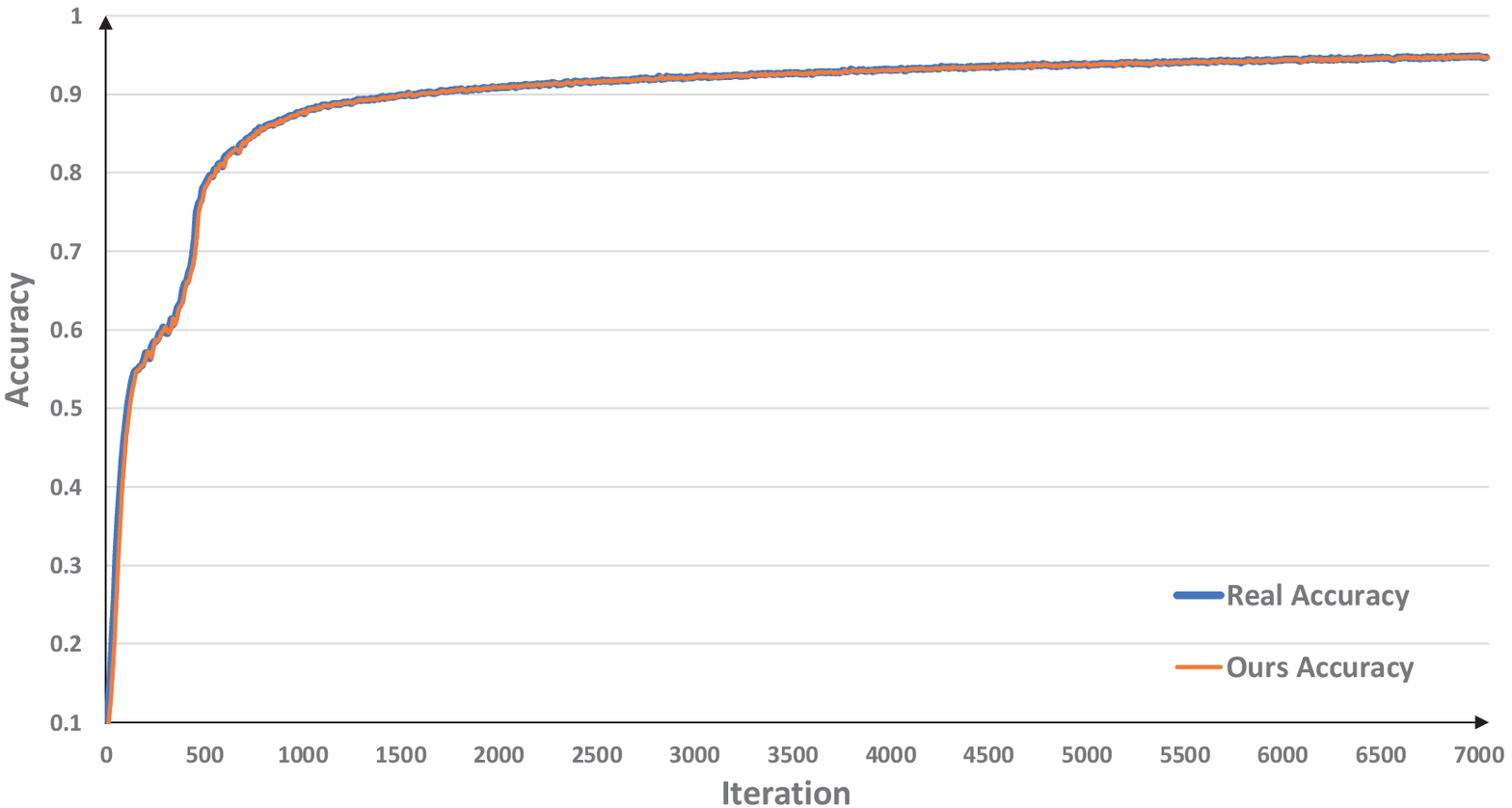}
\caption{The accuracy curves of secure training and \\plaintext training of Network-A}
\label{fig4}
\end{minipage}%
\begin{minipage}[c]{0.5\textwidth}
\centering
\includegraphics[height=5cm,width=8.5cm]{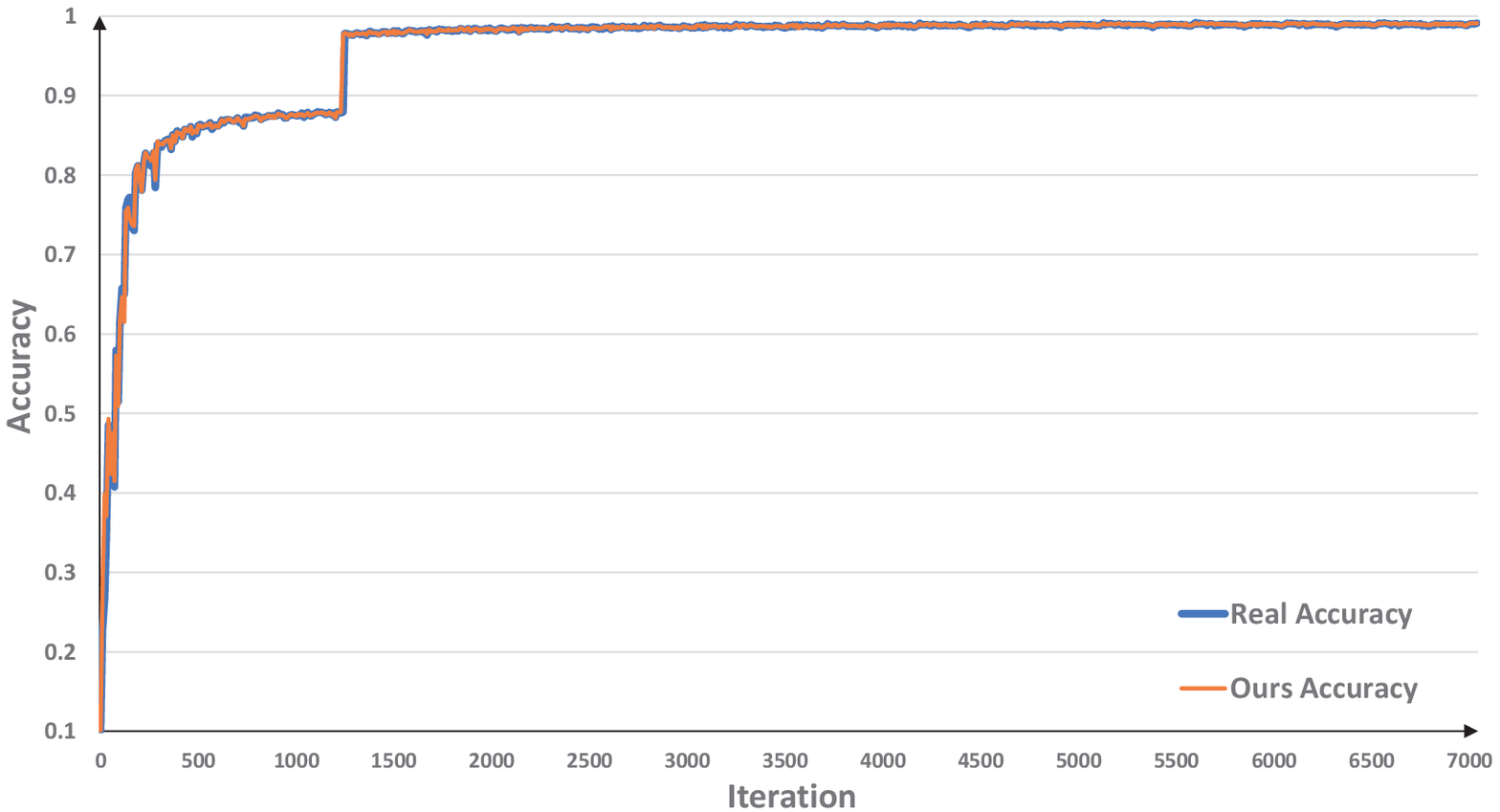}
\caption{The accuracy curves of secure training and \\plaintext training of Network-C}
\label{fig5}
\end{minipage}
\end{figure*}

We conduct 15 epoch trainings on four network structures and obtain the corresponding accuracy curves. Due to space limitations, we only list the plaintext training results and the corresponding secure training results of Network-A and Network-C as shown in\textbf{ Figure \ref{fig4}} and\textbf{ Figure \ref{fig5}}, we can see that the trend of our secure training result curve and the plaintext curve are approximatively the same, and there is no obvious fluctuation. In the \textbf{TABLE \ref{tab2}}, we compare our results with SecureNN and plaintext training. From the data, compared with the plaintext training, the difference between the accuracies is only approximately $\pm0.05\%$.  It can be seen that since we no longer use ReLU to approximate the Softmax function, but use the exponential function approximated by the piecewise function, and the new division method to achieve the Softmax function based entirely on Secret Sharing, the accuracy of our trained model is further improved compared with the previous work. For Network-A, in SecureML and SecureNN, the accuracy obtained in the first 15 epochs only reaches 93.4\%, whereas the accuracy we obtained reaches 94.8\%, which is similar to the accuracy of plaintext training. For Network-B, SecureNN only reached 98.77\%, but we reached 99.09\%. Network-D, which SecureNN does not train, we train it and obtain an accuracy similar to that of plaintext training.

\subsubsection{Efficiency Comparison}

\begin{figure*}
\centering
\begin{minipage}[c]{0.33\textwidth}
\centering
\includegraphics[height=4cm,width=5cm]{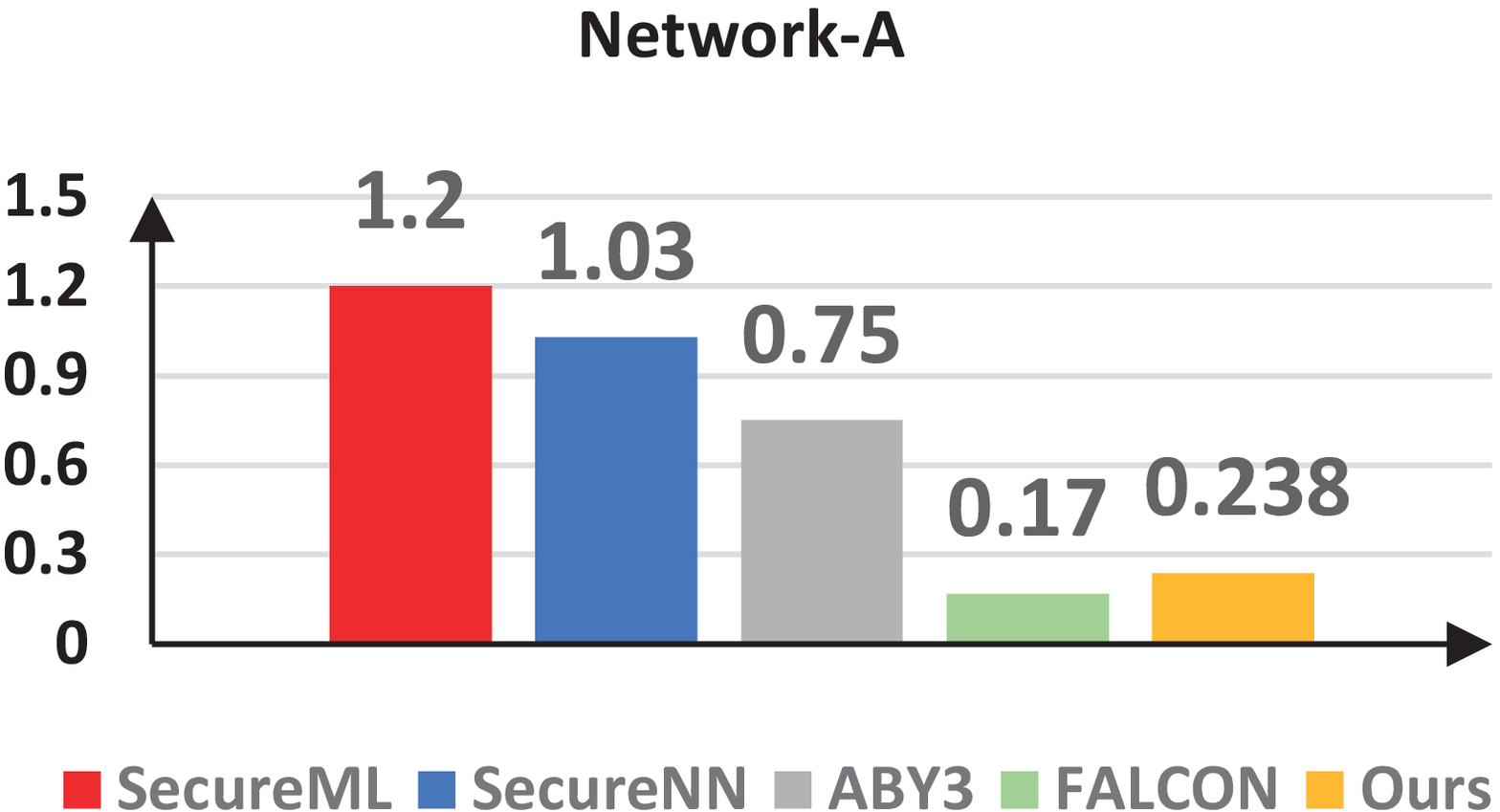}
\label{fig6}
\end{minipage}%
\begin{minipage}[c]{0.33\textwidth}
\centering
\includegraphics[height=4cm,width=5cm]{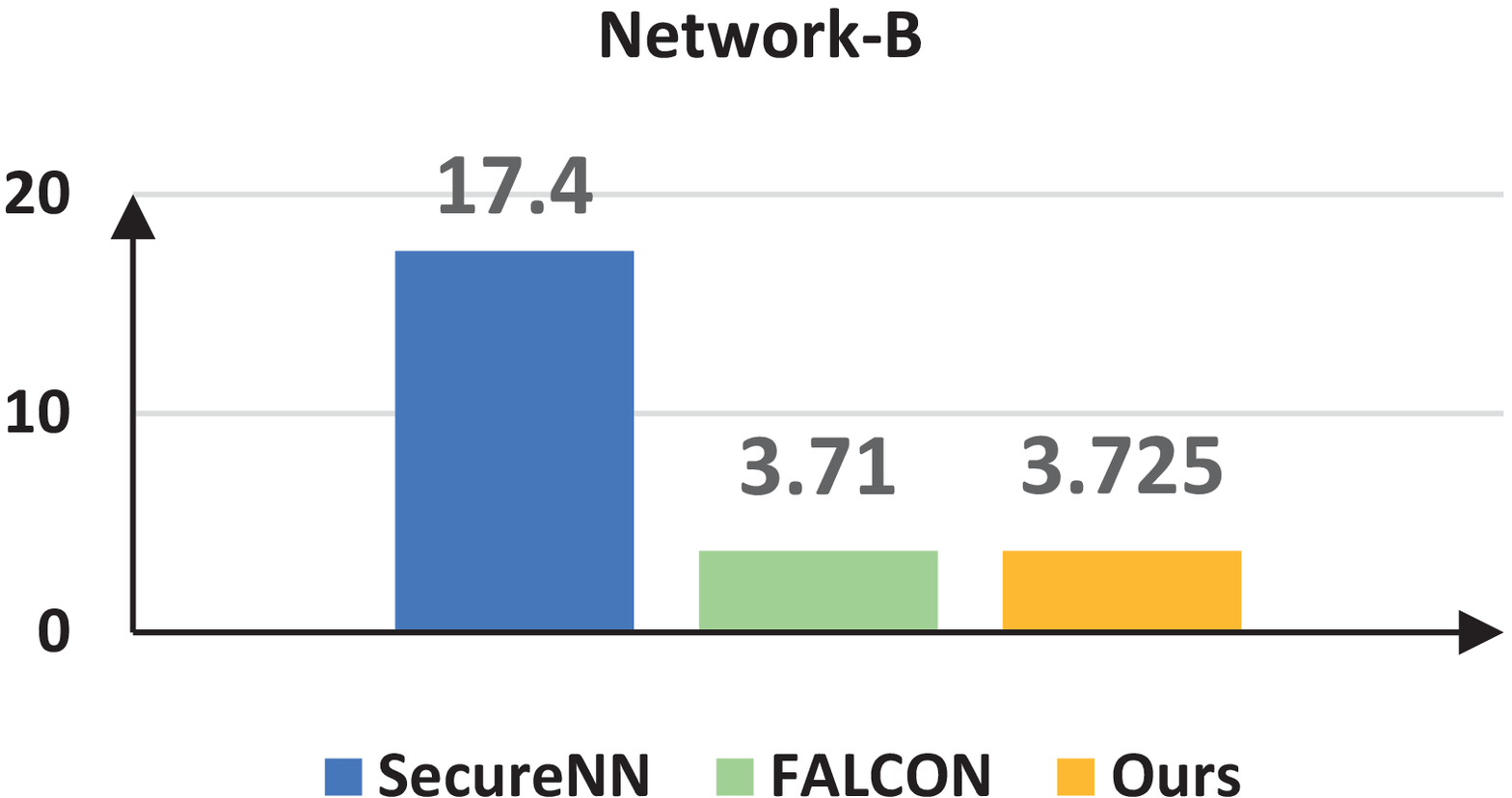}
\end{minipage}%
\begin{minipage}[c]{0.33\textwidth}
\centering
\includegraphics[height=4cm,width=5cm]{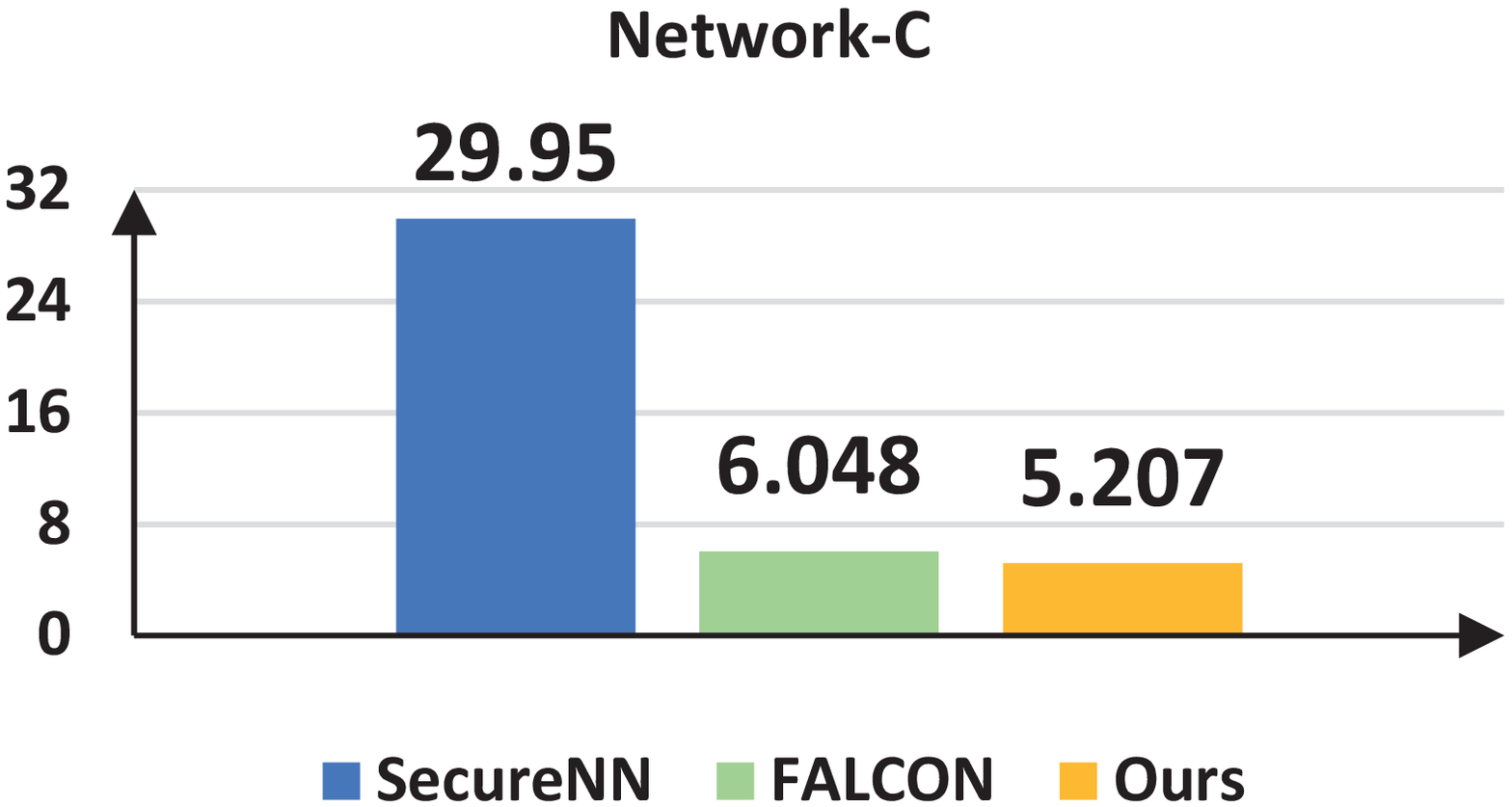}
\end{minipage}
\caption{Time needed for different frameworks in Network-A,B,C in the LAN setting. All runtimes are reported in hours.}
\label{fig6}
\end{figure*}

\begin{figure*}
\centering
\begin{minipage}[c]{0.33\textwidth}
\centering
\includegraphics[height=4cm,width=5cm]{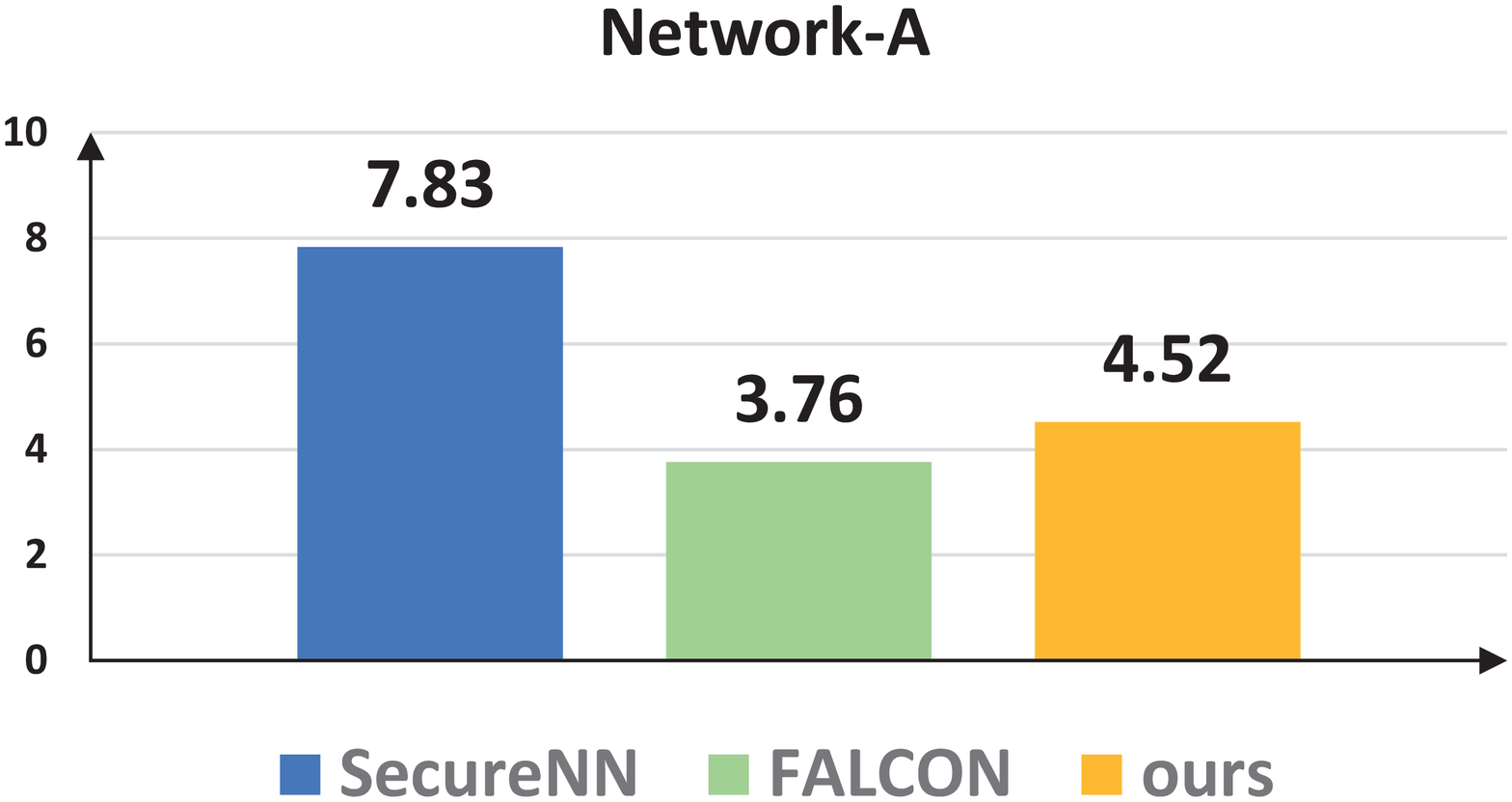}
\end{minipage}%
\begin{minipage}[c]{0.33\textwidth}
\centering
\includegraphics[height=4cm,width=5cm]{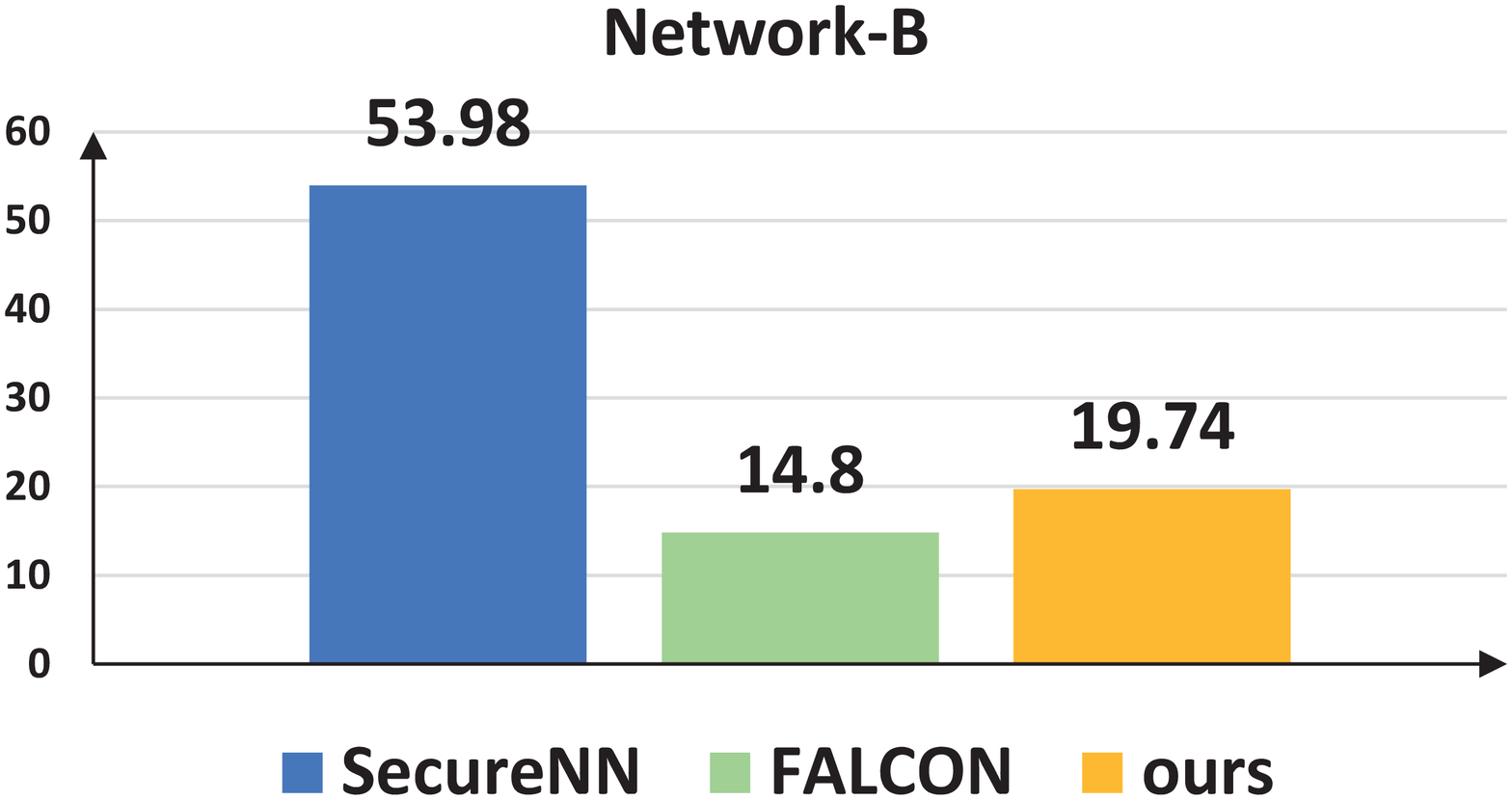}
\end{minipage}%
\begin{minipage}[c]{0.33\textwidth}
\centering
\includegraphics[height=4cm,width=5cm]{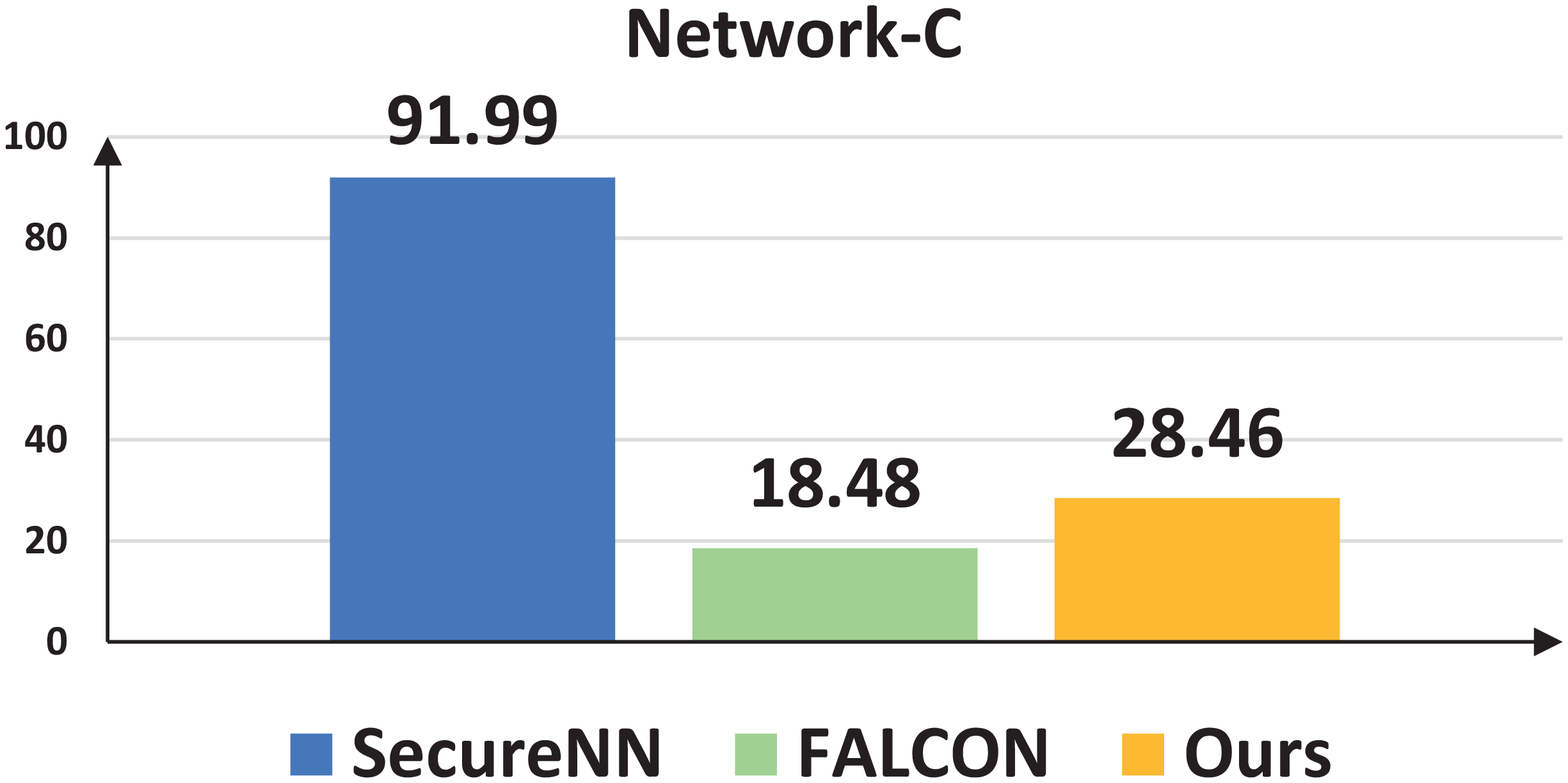}
\end{minipage}
\caption{Time needed for different frameworks in Network-A,B,C in the WAN setting. All runtimes are reported in hours.}
\label{fig7}
\end{figure*}

\begin{table}[]
\centering
\caption{Comparison of the training communication of the three networks. All communication is reported in TB.}
\setlength{\tabcolsep}{5mm}{
\renewcommand{\arraystretch}{1.5}
\begin{tabular}{|c|c|c|c|}
\hline
\multirow{2}{*}{Framework} & \multicolumn{3}{c|}{Communication} \\ \cline{2-4} 
                           & Network-A   & Network-B   & Network-C  \\ \hline
ABY3                       & 0.031       &             &            \\ \hline
SecureNN                   & 0.11        & 30.6        &            \\ \hline
FALCON                     & 0.016       & 0.54        & 0.81       \\ \hline
Ours                       & 0.047       & 1.157       & 1.642      \\ \hline
\end{tabular}}
\label{tab3}
\end{table}

In this part, we give a comparison of the secure training time and total communication of Network-A, B, C in the LAN and WAN setting with the previous works as shown in \textbf{Figure \ref{fig6}}, \textbf{Figure \ref{fig7}} and \textbf{TABLE \ref{tab3}}. What we need to declare here is that for the framework that separates the offline and online phases, such as SecureML and our work, we only take the online phase for comparison. In comparison with other works, we believe that the time of the online phase and the offline phase should not be simply added, because the data in the offline phase and the online phase are not correlated. In practical applications, under the two-server model, the triplets that need to be pre-computed or the mask r required for comparison may have been generated in advance, or only a part needs be generated before the training starts, and then the offline phase and the online phase can be completed in parallel as long as the data generated in the offline phase meets the requirements of the online phase. Therefore, the total time required may only be a small part more than the online phase. For the convenience of theoretical comparison, only the online phase is compared here. In the LAN setting, for Network-A, our runtime is 5$ \times $ faster than SecureML, 4.32 $ \times $ faster than SecureNN, and 3.15$ \times $ faster than ABY$^3$. For Network-B and Network-C, our work is also 4.67-5.75$ \times $ faster than SecureNN. And our work is very close to the current best 3PC implementation, FALCON. In the WAN setting, we also get approximate results. In terms of communication, our work is significantly reduced compared to SecureNN. For Network-A, it is about 42\% of its, and for Network-B it is even more obvious, which is about 3.8\%. It is very close to FALCON.

\subsection{Secure Inference}

\begin{table}[]
\centering
\caption{Comparison of time and communication in the three networks with previous works. All runtimes are reported in seconds and communication is MB.}
\setlength{\tabcolsep}{5mm}{
\renewcommand{\arraystretch}{1.5}
\begin{tabular}{|c|c|c|c|c|c|c|}
\hline
\multirow{2}{*}{Framework} & \multicolumn{2}{c|}{Network-A} & \multicolumn{2}{c|}{Network-B} & \multicolumn{2}{c|}{Network-D} \\ \cline{2-7} 
                           & Time           & Comm.         & Time           & Comm.         & Time           & Comm.         \\ \hline
SecureML                   & 0.18           & -             & -              & -             & -              & -             \\ \hline
DeepSecure                 & -              & -             & 9.67           & 791           & -              & -             \\ \hline
EzPC                       & 0.4            & 76            & 5.1            & 501           & 0.6            & 70            \\ \hline
Gazelle                    & 0.03           & 0.5           & 0.33           & 70            & 0.05           & 2.1           \\ \hline
MiniONN                    & 0.14           & 12            & 5.74           & 636.6         & 0.4            & 44            \\ \hline
XONN                       & 0.13           & 4.29          & 0.15           & 32.1          & 0.16           & 38.3          \\ \hline
Ours                       & 0.005          & 0.94          & 0.027          & 1.08          & 0.011          & 0.88          \\ \hline
\end{tabular}}
\label{tab4}
\end{table}

In this part, we give a comparison of the time and communication required for secure inference of a single data in Network-A, B, and D in the LAN setting with the previous works based on 2PC as shown in \textbf{TABLE \ref{tab4}}. Same as the time comparison of secure training, we also compare the online time. The reason for only comparing online time is more obvious in secure inference, because the inference emphasizes more on response time. After the model is placed on the two servers, in order to improve the response time, both parties will definitely try to generate as many triplets and mask $r$ as possible when the client does not use and do not occupy the online time, so we only compare the online time. Through the time comparison in the table, our work should be the current best job on 2PC. For Network-A, it is 6-80 $ \times $ faster than other works, for Network-D, it is 4-54$ \times $ faster than others, and for Network-B, it is even more obvious, which is 5-358$ \times $ faster than others.

\section{Conclusion}

In this work, we design a faster and more accurate neural network training and inference framework based on 2PC. We build a new preprocessing protocol for mask generation, support and realize secret sharing comparison on 2PC, propose a new method to further reduce the communication rounds, and construct some building blocks based on the comparison protocol, such as division and exponential. We obtain a higher degree of approximation Softmax function and then realize the neural network training and inference process entirely based on the secret sharing method. The experimental results show that our work is superior to most current frameworks in terms of accuracy and time efficiency. In the four network structures most used in previous works, the accuracies we obtained are all higher than those of other works and are closer to the results of plaintext training. In terms of time efficiency, our work has significantly improved compared with others both on secure training and secure inference.

\section*{Acknowledgment}

This work is supported by the National Natural Science Foundation of China under Grant No.62072208 and No.61772229.

\bibliographystyle{cas-model2-names}

\bibliography{ref}

\appendix

\section{Functionality Descriptions}

\begin{figure*}[h]
\centering
\begin{minipage}{0.5\textwidth}
\fbox{%
\parbox{0.95\textwidth}{%
\begin{center}
    ${{\cal F}_{GenerateR}}$
\flushleft{\quad \textbf{Input:} \quad The functionality receives no inputs}
\flushleft{\quad \textbf{Output:}  \quad  Compute the 
following:}
\flushleft{\qquad 1.\quad Choose random number $r$, and get the \\ 
           \qquad  vector  of each bit of $r$,$\left\{ {r\left[ j \right]} \right\}$. }
\flushleft{\qquad 2.\quad Generate random shares of $r$, random \\
           \qquad  shares of  $\left\{ {r\left[ j \right]} \right\}$ over ${Z_p}$, and send back to  \\ \qquad the parties.}  
\end{center}
  }%
}
\caption{Ideal functionality for ${\Pi_{GenerateR}}$}
\label{fig8}
\end{minipage}%
\begin{minipage}{0.5\textwidth}
\fbox{%
\parbox{0.95\textwidth}{%
\begin{center}
             ${{\cal F}_{GetWrapped}}$
\flushleft{\quad \textbf{Input:} \quad The functionality receives inputs ${\left\langle r \right\rangle ^L}$}
\flushleft{\quad \textbf{Output:}  \quad  Compute the 
following:}
\flushleft{\qquad 1.\quad Compute $b = wrap\left( {{r_0},{r_1},L} \right)$}
\flushleft{\qquad 2.\quad Generate random shares of $b$, and send \\
                    \quad \quad back to the parties.\\
                    \quad }  
\end{center}
  }%
}
\caption{Ideal functionality for ${\Pi_{GetWrapped}}$}
\label{fig9}
\end{minipage}%
\end{figure*}

\begin{figure*}[h]
\centering
\begin{minipage}{0.5\textwidth}
\fbox{%
\parbox{0.95\textwidth}{%
\begin{center}
             ${{\cal F}_{CheckZero}}$
\flushleft{ \textbf{Input:} The functionality receives inputs $\left\{ {{{\left\langle {c\left[ j \right]} \right\rangle }^p}} \right\}$}
\flushleft{ \textbf{Output:}  Compute the 
following:}
\flushleft{ \quad 1.\quad Reconstruct $\left\{ {c{{\left[ j \right]}^p}} \right\}$, get $\eta$ if there is 0 in it.}
\flushleft{ \quad 2.\quad Generate random shares of $\eta$, and send back \\
            \quad  to the parties.}  
\end{center}
  }%
}
\caption{Ideal functionality for ${\Pi_{CheckZero}}$}
\label{fig10}
\end{minipage}%
\begin{minipage}{0.5\textwidth}
\fbox{%
\parbox{0.95\textwidth}{%
\begin{center}
             ${{\cal F}_{Compare}}$
\flushleft{\quad \textbf{Input:} \quad The functionality receives inputs ${\left\langle a \right\rangle ^L}$}
\flushleft{\quad \textbf{Output:}  \quad  Compute the 
following:}
\flushleft{ \qquad 1.\quad Reconstruct $a$ and get  $\upsilon$ if $a > 0$. }
\flushleft{ \qquad 2.\quad Generate random shares of $\upsilon$, and send \\
            \qquad back  to the parties.}  
\end{center}
  }%
}
\caption{Ideal functionality for ${\Pi_{Compare}}$}
\label{fig11}
\end{minipage}%
\end{figure*}

\begin{figure*}[h]
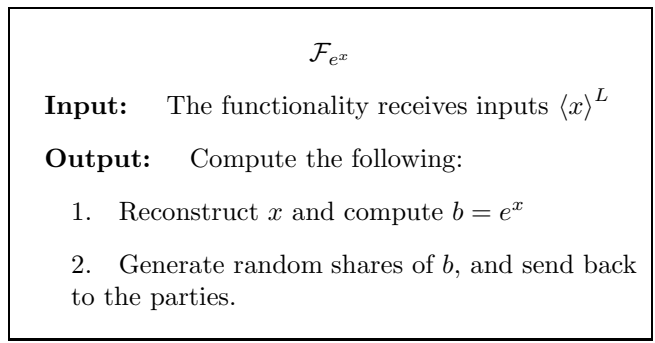

\centering
\begin{minipage}{0.5\textwidth}
\fbox{%
\parbox{\textwidth}{%
\begin{center}
             ${{\cal F}_{e^x}}$
\flushleft{\quad \textbf{Input:} \quad The functionality receives inputs ${\left\langle x \right\rangle ^L}$}
\flushleft{\quad \textbf{Output:}  \quad  Compute the 
following:}
\flushleft{\qquad 1.\quad Reconstruct $x$ and compute $b = {e^x}$}
\flushleft{\qquad 2.\quad Generate random shares of $b$, and send back \\
           \qquad to the parties.}  
\end{center}
  }%
}
\caption{Ideal functionality for ${\Pi_{e^x}}$}
\label{fig12}
\end{minipage}%
\begin{minipage}{0.5\textwidth}
\end{minipage}%
\end{figure*}

\end{document}